%% file: main.tex
\documentclass[letterpaper,12pt]{article}

\usepackage[margin=1in]{geometry}

\usepackage{amssymb}
\usepackage{amsmath}
\usepackage{amsthm}
\usepackage{verbatim}
\usepackage{enumerate}
\usepackage{paralist}
\usepackage[noend]{algorithmic}
\algsetup{indent=1cm}
\usepackage{algorithm}
\usepackage{comment}
\usepackage{cite}
\usepackage{bookmark}
\usepackage{mathtools}
\usepackage{xcolor}
\usepackage{authblk}

\newtheorem{theorem}{Theorem}
\newtheorem{lemma}{Lemma}
\newtheorem{definition}{Definition}
\newtheorem{corollary}{Corollary}
\newtheorem{proposition}{Proposition}

\newcommand{\eps}{\varepsilon}

\DeclareMathOperator\OPT{OPT}
\DeclareMathOperator\ALG{ALG}
\DeclareMathOperator\LQD{LQD}

\def\A{{\mathcal{A}}}
\def\B{{\mathcal{B}}}

\newcommand\Q[1]{\ensuremath{\mathcal{Q}^{(#1)}}}
\def\oopt{\textsc{LateQD}}

\DeclarePairedDelimiter\set{\lbrace}{\rbrace}

\DeclarePairedDelimiter\ceil{\lceil}{\rceil}
\DeclarePairedDelimiter\floor{\lfloor}{\rfloor}

\newcommand{\map}[2]{\set{#1 \mid #2}}

\def\myscalefigone{1.1}
\def\myscalefigtwo{1.1}
\def\myscale{1.}

\title{New Competitiveness Bounds for the Shared Memory Switch\thanks{This work was supported by the Russian Science Foundation grant no.~17-11-01276.}}

\author[1,2]{Ivan Bochkov}
\author[2,3]{Alex Davydow}
\author[1,2]{\\Nikita Gaevoy}
\author[2,3]{Sergey I. Nikolenko}
\affil[1]{St. Petersburg State University, St. Petersburg, Russia}
\affil[2]{Steklov Institute of Mathematics at St. Petersburg, Russia}
\affil[3]{Neuromation OU, Tallinn, Estonia}

\usepackage[EULERGREEK]{sansmath}
\usepackage{ifthen}
\usepackage{mathdots}
\usepackage{stmaryrd}
\usepackage{tikz}
\usetikzlibrary{plotmarks,shapes,snakes,calc,positioning,shadows.blur,decorations.pathreplacing,shapes.misc}
\usepackage{etoolbox}
\usepackage{pgfplots}
\usepackage{pgfkeys,pgffor}

\tikzset{%
        brace/.style = { decorate, decoration={brace, amplitude=5pt} },
       mbrace/.style = { decorate, decoration={brace, amplitude=5pt, mirror} },
        label/.style = { black, midway, scale=0.5, align=center },
     toplabel/.style = { label, above=.5em, anchor=south },
    leftlabel/.style = { label,rotate=-90,left=.5em,anchor=north },   
  bottomlabel/.style = { label, below=.5em, anchor=north },
        force/.style = { rotate=-90,scale=0.4 },
        round/.style = { rounded corners=2mm },
       legend/.style = { right,scale=0.4 },
        nosep/.style = { inner sep=0pt },
   generation/.style = { anchor=base },
          tbl/.style = { font=\bfseries\boldmath }
}

\colorlet{color02}{red!25!white}
\colorlet{color01}{orange!25!white}
\colorlet{color03}{green!25!white}
\colorlet{color04}{blue!25!white}
\colorlet{color04a}{blue!5!white}
\colorlet{color05}{magenta!25!white}
\colorlet{color06}{black!15!white}
\colorlet{color07}{white}
\colorlet{color2}{red!80!black}
\colorlet{color2a}{red!50!white}
\colorlet{color1}{orange!80!black}
\colorlet{color1a}{orange!50!white}
\colorlet{color3}{green!50!black}
\colorlet{color3a}{green!90!black}
\colorlet{color4}{blue!80!black}
\colorlet{color4a}{blue!50!white}
\colorlet{color5}{magenta!80!black}
\colorlet{color5a}{magenta!50!white}
\colorlet{color6}{black!80!black}
\colorlet{color6a}{black!30!white}
\colorlet{color7}{black!80!black}
\colorlet{color7a}{black!50!white}
\colorlet{color8}{black!40!white}
\colorlet{color9}{yellow!50!black}

\newcommand\tikzpacketnode[2]{%
  \begin{tikzpicture}[x=1cm, y=1cm, scale=0.9, font=\sffamily\sansmath]
    \draw[fill=#1,blur shadow={shadow blur steps=5}] (0.05,0) -- (0.45,0)
        arc (90:0:.5mm) -- (0.5,-0.45) arc (0:-90:.5mm) -- (0.05,-.5)
        arc (-90:-180:.5mm) -- (0,-0.05) arc(180:90:.5mm) -- cycle;
    \ifstrempty{#2}{}{\node at(0.25,-0.25)  [nosep,align=center,text centered,scale=0.8] {$\boldsymbol{#2}$};}
  \end{tikzpicture}
}

\newcommand\tikzpacket[4]{%
    \node at (#1, #2) {\tikzpacketnode{#3}{#4}};
}

\def\tbly{0.55}
\def\tblx{0.55}
\def\tblxstart{0.0}
\def\tblystart{0.0}

\newcommand\tikzpacketcirc[5]{%
    \node at (#1, #2) {\tikzpacketnode{#3}{#4}};
    \pgfmathsetmacro\x{#1-(.3)*\tblx}
    \pgfmathsetmacro\y{#2-(.25)*\tbly}
    \node[circle,draw=black,fill=blue!20!white,inner sep=1.2pt] at (\x, \y) {{\scriptsize\bfseries #5}};
}

\newcommand\tikzroundedrect[5]{%
    \pgfmathsetmacro\x{(#1-.65)*\tblx}
    \pgfmathsetmacro\y{(#2-.65)*\tbly}
    \pgfmathsetmacro\xx{(#3+.55)*\tblx}
    \pgfmathsetmacro\yy{(#4+.7)*\tbly}
    \draw[rounded corners] (\x, \y) rectangle (\xx, \yy);
    \node[anchor=north east] at (\xx, \yy) {#5};
}

\newcommand\tbltext[3]{%
    \pgfmathsetmacro\x{\tblxstart+(#1)*\tblx}
    \pgfmathsetmacro\y{\tblystart+(#2)*\tbly}
    \node at (\x, \y) [nosep,align=center,text centered,scale=0.8] {#3};
}

\newcommand\tblp[4]{%
    \pgfmathsetmacro\x{\tblxstart+(#1)*\tblx}
    \pgfmathsetmacro\y{\tblystart+(#2)*\tbly}
    \tikzpacket{\x}{\y}{#3}{#4};
}

\newcommand\tblc[5]{%
    \pgfmathsetmacro\x{\tblxstart+(#1)*\tblx}
    \pgfmathsetmacro\y{\tblystart+(#2)*\tbly}
    \tikzpacketcirc{\x}{\y}{#3}{#4}{#5};
}

\newcommand\tblhseq[4]{%
 \foreach \x in {#1,...,#2} {
    \tblp{\x}{#3}{#4}{};
 }
}

\begin{document}

\maketitle

\begin{abstract}
We consider one of the simplest and best known buffer management architectures: the shared memory switch with multiple output queues and uniform packets. It was one of the first models studied by competitive analysis, with the Longest Queue Drop (LQD) buffer management policy shown to be at least $\sqrt{2}$- and at most $2$-competitive; a general lower bound of $4/3$ has been proven for all deterministic online algorithms. Closing the gap between $\sqrt{2}$ and $2$ has remained an open problem in competitive analysis for more than a decade, with only marginal success in reducing the upper bound of $2$. In this work, we first present a simplified proof for the $\sqrt{2}$ lower bound for LQD and then, using a reduction to the continuous case, improve the general lower bound for all deterministic online algorithms from $\frac 43$ to $\sqrt{2}$. Then, we proceed to improve the lower bound of $\sqrt{2}$ specifically for $\LQD$, showing that $\LQD$ is at least $1.44546086$-competitive. We are able to prove the bound by presenting an explicit construction of the optimal clairvoyant algorithm which then allows for two different ways to prove lower bounds: by direct computer simulations and by proving lower bounds via linear programming. The linear programming approach yields a lower bound for $\LQD$ of $1.4427902$ (still larger than $\sqrt{2}$).

\end{abstract}


\input{intro}

\input{related}

\input{model}

\input{OPT}

\input{extensions}

\input{lqd}

\input{general}

\input{linprog}

\input{eval}

\input{conclusion}

\bibliographystyle{plain}
\bibliography{np}

\end{document}

%% file: intro.tex
\section{Introduction}\label{sec:intro}

Buffering architectures define how input and output ports of a network element are connected, and buffer management policies determine how individual packets get processed. While classical works on buffer management used stochastic models to evaluate the performance of various policies, in modern networking a network edge has to process increasingly diverse and unpredictable incoming traffic, which leads to the need for worst-case guarantees.

Such guarantees can be provided by \emph{competitive analysis}, an approach originally applied to the analysis of online algorithms in the 1980s~\cite{Borodin-ElYaniv} but since the early 2000s increasingly used to study buffer management policies. An online algorithm is said to be \emph{$\alpha$-competitive} if for any input (any possible incoming sequence of packets) it achieves total value at least $\alpha$ times less than what a clairvoyant offline algorithm could achieve on the same sequence. Competitive analysis allows to obtain worst-case guarantees: an upper bound on competitiveness means that an algorithm does not lose too much on \emph{any} input sequence. Over the last two decades, lower and upper bounds on the competitiveness of various buffer management policies have been proven in many different settings; for detailed surveys of the field we refer to~\cite{NK14,G10,ES04}.

One of the foundational works that introduced competitive analysis into buffer management was the work by Hahne, Kesselman, and Mansour~\cite{HK+01}, later extended by Aiello, Kesselman, and Mansour~\cite{AKM08}. These works considered one of the simplest nontrivial buffering architectures: a shared memory switch with multiple output queues and uniform (identical) packets. Incoming packets in this model are destined to one of the several output queues, which share a common buffer of finite size $B$. A buffer management policy for the shared memory switch must decide which queues to push packets out from when the buffer overflows, with the purpose of achieving maximal throughput (equivalently, dropping as few packets as possible).

For the shared memory switch, the works~\cite{HK+01,AKM08} introduced a very natural online algorithm, the Longest Queue Drop (LQD). It pushes packets out of the longest queues, trying to equalize queue sizes and thus keep as many queues as possible transmitting packets at the same time. As for the competitiveness of LQD, the works~\cite{HK+01,AKM08} showed a lower bound of~$\sqrt{2}$ and an upper bound of $2$. Another interesting result was a general lower bound of $4/3$ established for the competitive ratio of any deterministic online algorithm.

Since then, the problem of closing the gap between $\sqrt{2}$ and $2$ for LQD, as well as between $4/3$ and $2$ for all online algorithms, has been one of the key open problems in theoretical analysis of buffer management policies. In particular, it was listed as an important open problem in a \emph{SIGACT News} survey by Goldwasser~\cite{G10}. So far the only new result in this specific setting has been provided by Kobayashi, Miyazaki, and Okabe~\cite{Kobayashi:2008:TBO:1522914.1522915}, who improved the upper bound to $2-\frac{1}{B}\min_{i=1}^N\left\{\left\lfloor\frac{B}{i}\right\rfloor+i-1\right\}$, where $N$ is the number of output queues and $B$ is the size of the buffer. This bound tends to $2$ as $B$ and $N$ tend to infinity, but it still shows an important point: it has turned out that $2$ is not a crucial number for this case, and it can be potentially improved.

In this work, we introduce novel techniques for the analysis of online policies for the shared memory switch with uniform packets and make the next steps towards closing the gap between lower and upper bounds on their competitiveness. Our first result here is an explicit construction for the optimal offline clairvoyant algorithm, and the second is a generalized construction of a set of hard instances that we use in the lower bounds. Using a novel approach to proving lower bounds on competitive ratios through this construction, we present a simplified proof of the lower bound of $\sqrt{2}$ for LQD for a shared memory switch. Moreover, the new simplified proof generalizes well, which allows us to obtain the main result of this work: improve the \emph{general} lower bound presented in~\cite{AKM08} from $\frac 43$ to $\sqrt{2}$.

Note that once an efficient construction of the optimal algorithm has been found (ours is actually linear), one can look for hard instances by running computer simulations, comparing online policies such as $\LQD$ against the optimal algorithm. To this end, we develop a special form of input instances and find a representation of online algorithms such that the number of processed packets can be found as a solution of a linear programming problem. Moreover, we implement an explicit form of the optimal algorithm and $\LQD$ and run computer simulations on the presented hard instances. As a result, we obtain a lower bound on the competitiveness of $\LQD$ of $1.44546086$ (which is better than previously known $\sqrt{2}$) and a general lower bound of $1.32742316$ (worse even than previously known $\frac 43$ and, naturally, the newly proven $\sqrt{2}$) for any deterministic online algorithm based on these simulations. Thus, we are bringing the general lower bound to the former $\LQD$ bound of $\sqrt{2}$ and at the same time improving the $\LQD$ bound further.

The paper is organized as follows. Section~\ref{sec:related} surveys related work, and Section~\ref{sec:model} formally introduces the model. In Section~\ref{sec:opt} we present the construction of the optimal offline algorithm $\oopt$ and prove its optimality. Section~\ref{sec:extensions} introduces two important extensions to the model: fractional packets (Section~\ref{sec:frac}) and our construction of a family of hard instances (Section~\ref{sec:hard}). With the help of these extensions, Section~\ref{sec:lqd} presents a simplified proof of the $\sqrt{2}$ lower bound for $\LQD$. Section~\ref{sec:general} presents our main theoretic result: a general lower bound of $\sqrt{2}$ for any deterministic online algorithm. Section~\ref{sec:linprog} shows how to reduce finding the number of processed packets for buffer management algorithms to solving linear programming problems, Section~\ref{sec:eval} presents the results of our simulation study, and Section~\ref{sec:concl} concludes the paper.

%% file: related.tex
\section{Related work}\label{sec:related}

In this work, we consider the setting of a shared memory switch that receives identical incoming packets, each destined for one of the $N$ output queues that share a total memory of $B$. We have already outlined existing competitive analysis results for policies with push-out in Section~\ref{sec:intro}. The case of non-push-out policies, which make admission decisions but then are not allowed to drop already accepted packets, was studied in the work~\cite{KM04}, where a non-constant general lower bound of $\frac{\log N}{\log\log N}$ on the competitiveness of any online deterministic algorithm is presented together with a specific algorithm that achieves an upper competitiveness bound of $\ln N+2$.

Subsequent works considered other buffer management settings, extending either the buffering architecture (different configurations of input and output queues), packet characteristics (making packets non-uniform), or both.

The majority of works on packets with varying characteristics dealt with the \emph{values} of packets, i.e., the setting where a packet is characterized by a numerical value and the objective is to maximize the total transmitted value. For a single queue, the optimal competitive ratio for non-push-out policies was shown to be $\ln V$, with tightly matching bounds of $1+\ln V$ and $2+\ln V + O(\ln^2 V / B)$~\cite{AMZ03, Zhu04}, where $V$ is the maximal possible value of a packet. With push-out, the PQ policy (Priority Queue) that transmits largest values first and drops smallest values is obviously optimal. An important special case here is when the policy has to preserve FIFO ordering of packets; here, a general lower bound of $1.419$ was shown for all online algorithms~\cite{KMS05}, with a stronger bound of $1.434$ for $B=2$~\cite{AMZ03, Zhu04}, and the FIFO greedy push-out policy (push out the packet with smallest value) has been shown to be $2$-competitive~\cite{KLM04}. 

Another important characteristic is \emph{required processing}, when a packet needs from $1$ to $k$ time slots at the processor before it can be transmitted. In the single queue case, any greedy non-push-out policy is at least $\frac12(k+1)$-competitive, and for the push-out case PQ is again optimal~\cite{KKSS12,KLNS+12}. With the FIFO requirement, the class of so-called \emph{lazy} policies has been studied in~\cite{KLNS+12,KLNS12,KLNS16}, including the Lazy PQ policy shown to be $2$-competitive and a general upper bound of $\frac 1B\log_{\frac{B}{B-1}}k +3$ on the competitive ratio of any lazy policy
and a matching lower bound of $\frac 1B\log_{\frac{B}{B-1}}k +1$ for several processing orders.

Variable values and variable required processing have been coming together in recent works~\cite{CNK15,CNKD18}. In this case, priority queues are no longer trivially optimal, several different priorities are possible, and the results deal with competitive ratios of different priority queues. Finally, yet another important possible characteristic of a packet is the \emph{deadline}, or \emph{slack}, where a packet must be transmitted before a certain predefined time~\cite{DCNK17}.

As for architectures, another important case is the case of multiple separated queues, where each queue has its own memory, and a policy must choose which input queue to transmit a packet from and set admission policies for input queues. For the case of uniform packets, a buffer with independent queues, where each of $N$ input queues has a separate independent buffer of size $B$, has been considered in~\cite{AL06}. This is a rare case where the problem has been closed completely: the work~\cite{AL06} presents a deterministic policy with competitive ratio converging to $\frac{e}{e-1}\approx 1.582$ for arbitrary $B$, and a matching lower bound has been proven in~\cite{AR05}. The case of multiple separated queues for packets with variable values has been considered in~\cite{KKM12,AR05,KMO09}; for packets with variable required processing, in~\cite{KLNS13}, where a $2$-competitive policy was introduced for that setting, and in~\cite{EKNS14}, where each queue is constrained to contain packets with the same processing requirement.

Combined input-output queued (CIOQ) switches have been considered in~\cite{KesselmanR06,AzarR06,KesselmanKS12,KesselmanR08}, with constant competitiveness achieved for certain policies. Buffered crossbar switch architectures, with a small buffer on every crosspoint in addition to input and output queues, have been studied in~\cite{KesselmanKS12} for uniform packets and in~\cite{KKM+10} for packets with variable values, again with constant competitiveness bounds but in the latter case with rather large constants, about $19.95$~\cite{KKM+10}. The currently best results in this direction were provided by Al-Bawani et al.~\cite{Al-Bawani2018}, who presented a faster $3$-competitive algorithm for uniform packets, a $5.83$-competitive algorithm for CIOQ switch, and a $14.83$-competitive algorithm for a buffered crossbar switch.

Thus, at present competitive analysis is an important tool for studying buffer management policies in various settings, with many different architectures and packet characteristics defining the landscape of the field. However, only in rare cases this research has already produced tight matching bounds on the competitiveness of specific policies. In this work, we go back to the foundational works~\cite{HK+01,AKM08}  and one of the first and arguably simplest models, a shared memory switch with uniform packets and push-out, and improve upon the lower bound of $\sqrt{2}$ on the competitiveness of $\LQD$ and upon the general lower bound of $\frac 43$ for an arbitrary deterministic online algorithm.

%% file: model.tex
\section{Model description}\label{sec:model}

\subsection{Basic model}\label{sec:basic}

\input{fig1.tex}

Consider a shared memory switch with one input port and $N$ output ports; in~\cite{AKM08}, the number of input ports was used to bound the number of incoming packets, but for simplicity we assume here that a single input port can receive arbitrarily many packets on every time slot. Each output port has its own queue of arbitrary size with the constraint that the total size of all queues does not exceed $B$ (shared memory switch with memory size $B$). Packets are uniform, i.e., they all have the same value and size, and each packet $p$ is labeled with its desired output port $i\in\{1,\ldots,N\}$.

A sequence of packets arrives at discrete slotted time. Each time slot consists of two phases: 
\begin{enumerate}[(1)]
\item on the \emph{arrival} phase of a time slot $t$, there arrives a set of packets $\A_t$; the buffer management policy is free to accept all or any packets from $\A_t$ in the buffer and push out any subset of packets from the buffer; we denote the set of packets in the buffer at the beginning of time slot $t$ by $\B_t$ and after arrival by $\B'_t$, so $\B'_t\subseteq \B_{t}\cup \A_t$ and $|\B'_t|\le B$;
\item on the \emph{transmission} phase, each output queue $\Q i$, $i=1,\ldots,N$, sends out its head of line packet; since packets are uniform, it does not matter which one to send out, so we treat $\Q i$ as simply a set of packets or even just a number of packets; thus, after transmission we get ${\Q{i}}'$ with $|{\Q{i}}'|=\max\left(0, |\Q{i}|-1\right)$, and $\B_{t+1}=\cup_{i=1}^N{\Q{i}}'$.
\end{enumerate}
We illustrate the model on Fig.~\ref{fig:setting}. The figure shows output queues along the vertical axis ($N=3$ in this case) and packets along the horizontal axis. We also show the ``virtual buffer state'' before push-out, when the buffer may be overflowing and has to be cut back to $B$; pushed out packets are shaded in red. After push-out, $\LQD$ has $4$ packets in every queue, and each queue transmits one packet, leaving the buffer state as shown on the right.

An input instance, or simply \emph{instance}, is a sequence of arrivals $\tau:{\mathbb N}\to\{1,\ldots,B\}^N$ that defines how many packets arrive at which destinations on every time step; an instance has \emph{duration} $T(\tau)$ if no more packets arrive after time $T$, i.e., $\forall t>T(\tau)\ \tau(t)=(0,\ldots,0)$. For every instance $\tau$ of finite duration and every online algorithm $\ALG$, $\ALG(\tau)$ is the total number of packets transmitted by $\ALG$ on input $\tau$ (over all time, i.e., up until at most time $T(\tau)+B$).

We denote by $\OPT$ the optimal offline algorithm that knows the entire sequence in advance and has infinite computational capacity. The \emph{competitive ratio} of an algorithm $\ALG$ is defined as the ratio of the number of transmitted packets of the optimal algorithm and $\ALG$ on the worst possible input instance of finite duration:
$$\alpha_{\ALG} = \sup_{\sigma:T(\tau)<\infty}\alpha_{\ALG}(\tau),\text{ where }\alpha_{\ALG}(\tau)=\frac{\OPT(\tau)}{\ALG(\tau)}.$$
If the competitive ratio is infinite, it is usually of interest to find its asymptotics with respect to model parameters, in this case $B$ and $N$.

Consider an increasing sequence of instances $(\tau)$, where $\tau_n$ has length $n$, and the first $n$ time slots of $\tau_{n + 1}$ coincide with $\tau_n$, i.e., $\tau_1\subseteq\tau_2\subseteq\ldots\subseteq\tau_n\subseteq\ldots$ For such a sequence, we call $\tau=\cup_{n=1}^{\infty}\tau_n$ an instance of infinite duration, or \emph{infinite instance}, and define the competitive ratio of an algorithm $\ALG$ as an upper limit of competitive ratios of $\ALG$ on $\tau_n$ as $n\to\infty$: $\alpha_{\ALG}(\tau) = \varlimsup_{n\to\infty}\alpha_{\ALG}(\tau_n)$. Note that an infinite instance can use an infinite number of output queues.

The setting and sample operation of the $\LQD$ policy are illustrated on Fig.~\ref{fig:setting} for the case of $N=3$ queues and $B=12$ with initial buffer state $|\Q1|=2$, $|\Q2|=4$, $|\Q3|=6$. After the arrivals as shown on the left, packets shaded in red are pushed out, and the queues have $4$ packets each; then, during transmission every queue transmits one packet, leaving the buffer state as shown on the right, with $|\Q{i}|=3$ for every $i$.

%% file: fig1.tex
\begin{figure}[t]\centering
\scalebox{\myscalefigone}{
\begin{tikzpicture}[x=1.2cm, y=1.2cm, scale=0.8]
    \tbltext{-5}{2}{${\cal Q}_1$};
    \tbltext{-5}{1}{${\cal Q}_2$};
    \tbltext{-5}{0}{${\cal Q}_3$};

	\tbltext{-1.5}{4.2}{{Arrivals}};
	\tbltext{3.5}{4.2}{{Buffer before admission}};

	\tblhseq{-3}{-1}{0}{color07};
	\tblhseq{-4}{-1}{1}{color07};
	\tblhseq{-2}{-1}{2}{color07};
	
	\tblhseq{1}{6}{0}{color07};
	\tblhseq{3}{6}{1}{color07};
	\tblhseq{5}{6}{2}{color07};
	\tikzroundedrect{1}{0}{6}{3}{$\LQD$};

	\tbltext{12}{4.2}{{Virtual buffer during admission}};
	\tbltext{19}{4.2}{{After transmission}};

	\tblhseq{13}{16}{2}{color07};
	\tblhseq{9}{12}{1}{color02};
	\tblhseq{13}{16}{1}{color07};
	\tblhseq{8}{12}{0}{color02};
	\tblhseq{13}{16}{0}{color07};
	\tikzroundedrect{8}{0}{16}{3}{$\LQD$};

	\tblhseq{18}{20}{0}{color07};
	\tblhseq{18}{20}{1}{color07};
	\tblhseq{18}{20}{2}{color07};
	\tikzroundedrect{18}{0}{20}{3}{$\LQD$};
\end{tikzpicture}
}
    \caption{Sample operation of a shared memory switch with $N=3$ queues and $B=12$. After arrivals, the $\LQD$ policy pushes out packets shaded in red, equating the queue lengths.}
    \label{fig:setting}
\end{figure}
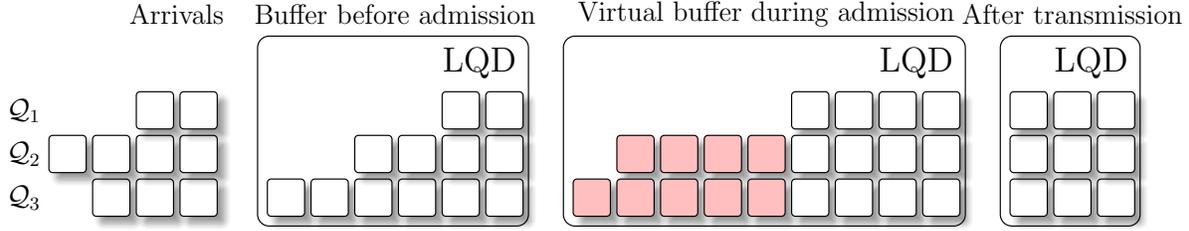

%% file: OPT.tex
\section{Explicit construction of the optimal algorithm}\label{sec:opt}

In this section, we present an explicit construction for the optimal clairvoyant algorithm for the shared memory switch. This is a new result in itself, and in subsequent sections it will prove to be important for studying competitive ratios of online algorithms, both in theory and in computational experiments.

Since packets are uniform, without loss of generality we assume that packets are processed and transmitted in the LIFO order, and they are pushed out in FIFO order, i.e., dropping a packet out of queue $i$ means pushing out the packet that arrived first.

Now for every packet we define a time interval during which it can reside in its corresponding queue. To do that, we simply run the buffer management algorithm with unlimited buffer $B=\infty$. In this case, under the assumptions above there are no decisions to be made (no packets ever need to be pushed out).
We denote by $[b_p,e_p]$ the time interval that packet $p$ spends in the buffer; note that due to the LIFO processing order it might happen that $e_p=\infty$; we denote $\Delta_p=e_p-b_p$ (also possibly infinite).

\begin{definition}
A buffer management algorithm is called \emph{regular} if it satisfies the following conditions:
    \begin{itemize}
        \item it is \emph{work-conserving}, i.e., each queue always processes and transmits packets when it is non-empty;
        \item it discards packets only when the buffer overflows.
    \end{itemize}
\end{definition}

\begin{lemma}\label{lem:regular}
For any buffer management algorithm $\ALG$, there exists a regular algorithm $\ALG'$ that has at least the same number of transmitted packets on every input sequence.
\end{lemma}
\begin{proof}
    We change $\ALG$ into $\ALG'$ packet by packet. Let $t$ be the first time slot when $\ALG$ becomes non-regular. If $\ALG$ has a queue $\Q i$ that is nonempty on time $t$ and does not transmit, let $\Q i$ transmit a packet; obviously, this does not increase memory requirements, and the new algorithm can operate exactly like $\ALG$ afterwards. If $\ALG$ pushed out a packet $p$ when its total occupied memory is less than $B$, let $p$ stay in the buffer and mark it to be the first one pushed out when the buffer overflows next time; obviously, the new algorithm again operates at least as well as $\ALG$.
\end{proof}

Thus, from now on we assume that all algorithms are work-conserving and push packets out only when necessary. This means that a buffer management algorithm in our setting (shared memory switch with uniform packets) is completely defined by the heuristic that chooses a queue to push out from when the buffer overflows. It turns out that the optimal algorithm is now easy to define.

\input{ex_lateqd.tex}

\begin{definition}\label{def:oopt}
When the buffer overflows, the (clairvoyant) algorithm $\oopt$ chooses the queue that contains the packet $p$ with the latest expected transmission moment $e_p$ (in particular, packets with $e_p=\infty$ are pushed out first).
\end{definition}

Figure~\ref{fig:lateqd} shows sample operation of $\oopt$ for a simple example with $N=3$ and $B=4$. The figure also shows, for comparison, the potential operation of an algorithm with infinite buffer and LIFO order of operation, which shows a clearer picture of where the packet priorities (marked as numbers on the packets) come from; the packet identities are labeled with letters inside the light blue circles.

\begin{theorem}\label{OPTisOPT}
    $\oopt$ is optimal.
\end{theorem}
\begin{proof}
Assume the opposite. Suppose there exists an input sequence $\tau$ on which $\oopt$ transmits strictly fewer packets than $\OPT$. Consider all algorithms that work optimally on $\tau$, and choose the algorithm $\ALG$ for which the first difference with $\oopt$ occurs the latest. By Lemma~\ref{lem:regular}, we can assume that $\ALG$ is regular.

Let $t_0$ be the moment of first difference between $\oopt$ and $\ALG$; note that this means that at $t_0$ the buffer overflows for both algorithms, and they make different push-out decisions. Consider the algorithm that makes this push-out decision according to $\oopt$ but otherwise makes the same decisions as $\ALG$. Since $\oopt$ chooses the packet with the latest expected transmission moment, the memory required for the new algorithm never exceeds $B$, and the new algorithm is feasible. Now this algorithm transmits the same number of packets but differs from $\oopt$ later, a contradiction.
\end{proof}

\input{ex_oopt}

The operation of $\oopt$ in comparison with $\LQD$ is illustrated on Fig.~\ref{fig:opt} for the case $N=3$, $B=6$. Numbers on the packets show the provisional expected transmission moments for incoming packets in LIFO order for an algorithm with infinite memory; clairvoyant $\oopt$ knows these numbers and uses them as packet priorities for push-out. We do not show packet identities here in order not to clutter the picture. Note that after $t=3$ $\oopt$ has $3$ packets in $\Q1$, and these packets will all be transmitted for a total of $6$ packets from $\Q1$, while $\LQD$ transmits only $4$ packets from $\Q1$. In total, on this input $\oopt$ transmits $1+2+3+3+3+2+1\times 5=19$ packets, and $\LQD$ transmits $1+2+3+3+2+1\times 6=17$ packets.




















%% file: ex_lateqd.tex
\def\myxlateqd{4}
\def\myxinf{14}

\renewcommand{\tblx}{0.75}
\renewcommand{\tbly}{0.7}

\newcommand\myboxessetuplateqd[3]{
    \tbltext{-4}{#3}{{\bf $t=#1$}};
    \tikzroundedrect{-1}{#2}{\myxlateqd}{#3}{$\oopt$};
    \tikzroundedrect{\myxinf-8}{#2}{\myxinf}{#3}{$B=\infty$};
}

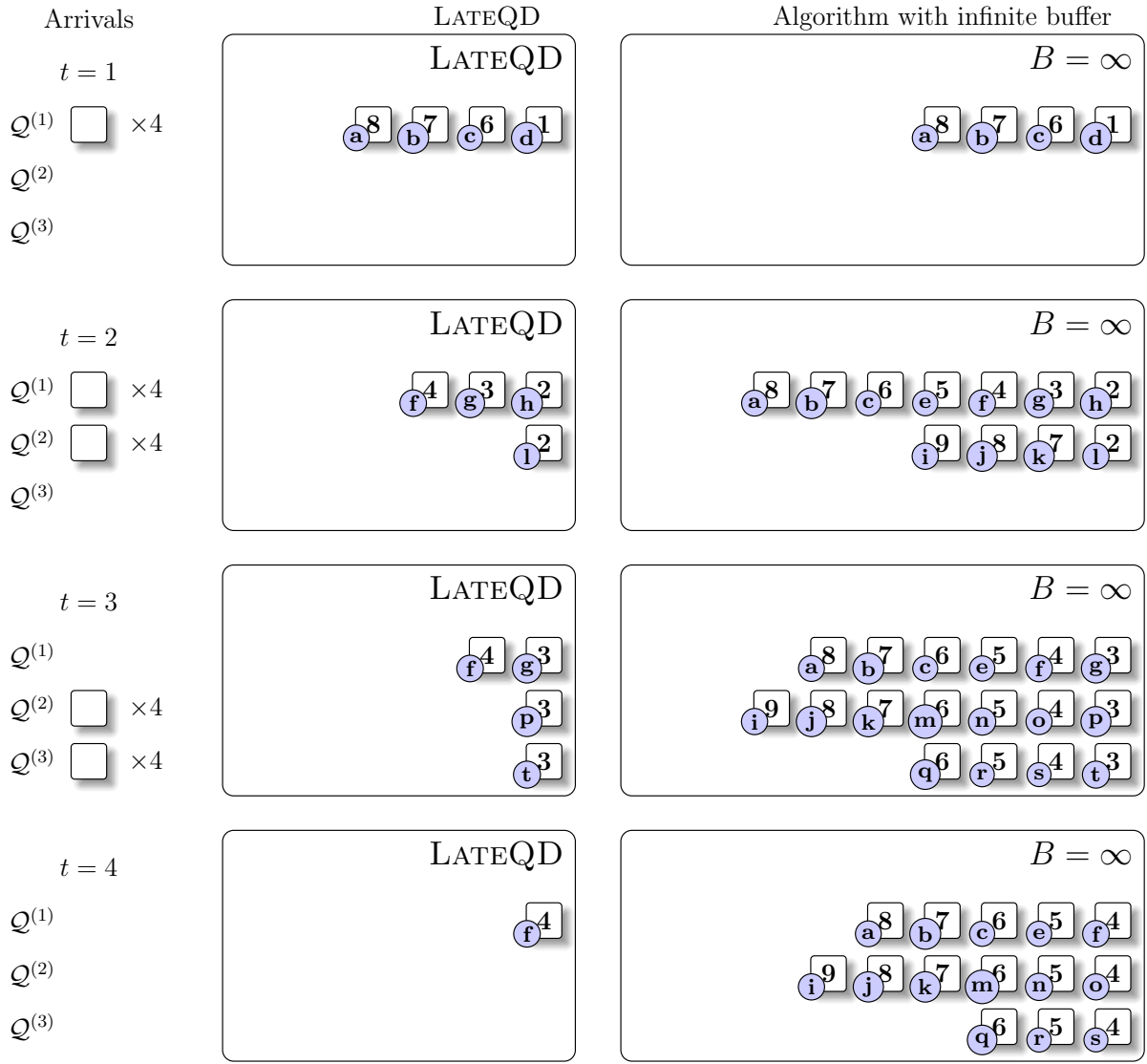
\begin{figure}[!t]\centering
\scalebox{\myscalefigtwo}{
\begin{tikzpicture}[x=1.2cm, y=1.2cm, scale=0.8]
    \tbltext{-5}{7}{$\Q 1$};
    \tbltext{-5}{6}{$\Q 2$};
    \tbltext{-5}{5}{$\Q 3$};
    \tbltext{-5}{2}{$\Q 1$};
    \tbltext{-5}{1}{$\Q 2$};
    \tbltext{-5}{0}{$\Q 3$};
    \tbltext{-5}{-3}{$\Q 1$};
    \tbltext{-5}{-4}{$\Q 2$};
    \tbltext{-5}{-5}{$\Q 3$};
    \tbltext{-5}{12}{$\Q 1$};
    \tbltext{-5}{11}{$\Q 2$};
    \tbltext{-5}{10}{$\Q 3$};

	\tbltext{-4}{14}{{Arrivals}};
    \tbltext{3}{14}{{$\oopt$}};
    \tbltext{11}{14}{{Algorithm with infinite buffer}};

    \tblp{-4}{12}{color07}{};
	\tbltext{-3}{12}{$\times 4$};

    \tblp{-4}{7}{color07}{};
    \tblp{-4}{6}{color07}{};
	\tbltext{-3}{7}{$\times 4$};
	\tbltext{-3}{6}{$\times 4$};

    \tblp{-4}{1}{color07}{};
    \tblp{-4}{0}{color07}{};
	\tbltext{-3}{1}{$\times 4$};
	\tbltext{-3}{0}{$\times 4$};
	
	\myboxessetuplateqd{1}{10}{13};
	\myboxessetuplateqd{2}{5}{8};
	\myboxessetuplateqd{3}{0}{3};
	\myboxessetuplateqd{4}{-5}{-2};



    \tblc{\myxinf-3}{12}{color07}{8}{a};
    \tblc{\myxinf-2}{12}{color07}{7}{b};
    \tblc{\myxinf-1}{12}{color07}{6}{c};
    \tblc{\myxinf}{12}{color07}{1}{d};

    \tblc{\myxinf-6}{7}{color07}{8}{a};
    \tblc{\myxinf-5}{7}{color07}{7}{b};
    \tblc{\myxinf-4}{7}{color07}{6}{c};
    \tblc{\myxinf-3}{7}{color07}{5}{e};
    \tblc{\myxinf-2}{7}{color07}{4}{f};
    \tblc{\myxinf-1}{7}{color07}{3}{g};
    \tblc{\myxinf}{7}{color07}{2}{h};

    \tblc{\myxinf-3}{6}{color07}{9}{i};
    \tblc{\myxinf-2}{6}{color07}{8}{j};
    \tblc{\myxinf-1}{6}{color07}{7}{k};
    \tblc{\myxinf}{6}{color07}{2}{l};

    \tblc{\myxinf-5}{2}{color07}{8}{a};
    \tblc{\myxinf-4}{2}{color07}{7}{b};
    \tblc{\myxinf-3}{2}{color07}{6}{c};
    \tblc{\myxinf-2}{2}{color07}{5}{e};
    \tblc{\myxinf-1}{2}{color07}{4}{f};
    \tblc{\myxinf}{2}{color07}{3}{g};

    \tblc{\myxinf-6}{1}{color07}{9}{i};
    \tblc{\myxinf-5}{1}{color07}{8}{j};
    \tblc{\myxinf-4}{1}{color07}{7}{k};
    \tblc{\myxinf-3}{1}{color07}{6}{m};
    \tblc{\myxinf-2}{1}{color07}{5}{n};
    \tblc{\myxinf-1}{1}{color07}{4}{o};
    \tblc{\myxinf}{1}{color07}{3}{p};

    \tblc{\myxinf-3}{0}{color07}{6}{q};
    \tblc{\myxinf-2}{0}{color07}{5}{r};
    \tblc{\myxinf-1}{0}{color07}{4}{s};
    \tblc{\myxinf}{0}{color07}{3}{t};

    \tblc{\myxinf-4}{-3}{color07}{8}{a};
    \tblc{\myxinf-3}{-3}{color07}{7}{b};
    \tblc{\myxinf-2}{-3}{color07}{6}{c};
    \tblc{\myxinf-1}{-3}{color07}{5}{e};
    \tblc{\myxinf}{-3}{color07}{4}{f};

    \tblc{\myxinf-5}{-4}{color07}{9}{i};
    \tblc{\myxinf-4}{-4}{color07}{8}{j};
    \tblc{\myxinf-3}{-4}{color07}{7}{k};
    \tblc{\myxinf-2}{-4}{color07}{6}{m};
    \tblc{\myxinf-1}{-4}{color07}{5}{n};
    \tblc{\myxinf}{-4}{color07}{4}{o};

    \tblc{\myxinf-2}{-5}{color07}{6}{q};
    \tblc{\myxinf-1}{-5}{color07}{5}{r};
    \tblc{\myxinf}{-5}{color07}{4}{s};

    \tblc{\myxlateqd-3}{12}{color07}{8}{a};
    \tblc{\myxlateqd-2}{12}{color07}{7}{b};
    \tblc{\myxlateqd-1}{12}{color07}{6}{c};
    \tblc{\myxlateqd}{12}{color07}{1}{d};

    \tblc{\myxlateqd-2}{7}{color07}{4}{f};
    \tblc{\myxlateqd-1}{7}{color07}{3}{g};
    \tblc{\myxlateqd}{7}{color07}{2}{h};
    \tblc{\myxlateqd}{6}{color07}{2}{l};

    \tblc{\myxlateqd-1}{2}{color07}{4}{f};
    \tblc{\myxlateqd}{2}{color07}{3}{g};
    \tblc{\myxlateqd}{1}{color07}{3}{p};
    \tblc{\myxlateqd}{0}{color07}{3}{t};

    \tblc{\myxlateqd}{-3}{color07}{4}{f};

\end{tikzpicture}
}
    \caption{Sample operation of $\oopt$ for $N=3$, $B=4$. Numbers on packets show provisional transmission moments for an algorithm with infinite buffer (shown on the right). Letters in light blue circles mark the identities of individual packets.
    }
    \label{fig:lateqd}
\end{figure}

\renewcommand{\tblx}{0.55}

%% file: ex_oopt.tex
\def\myxvirt{9}
\def\myxafter{15}
\def\myxlqd{19}

\newcommand\myboxessetup[3]{
    \tbltext{-4}{#3}{{\bf $t=#1$}};
    \tikzroundedrect{-1}{#2}{\myxvirt}{#3}{$\oopt$};
    \tikzroundedrect{\myxafter-4}{#2}{\myxafter}{#3}{$\oopt$};
    \tikzroundedrect{\myxlqd-2}{#2}{\myxlqd}{#3}{$\LQD$};
}

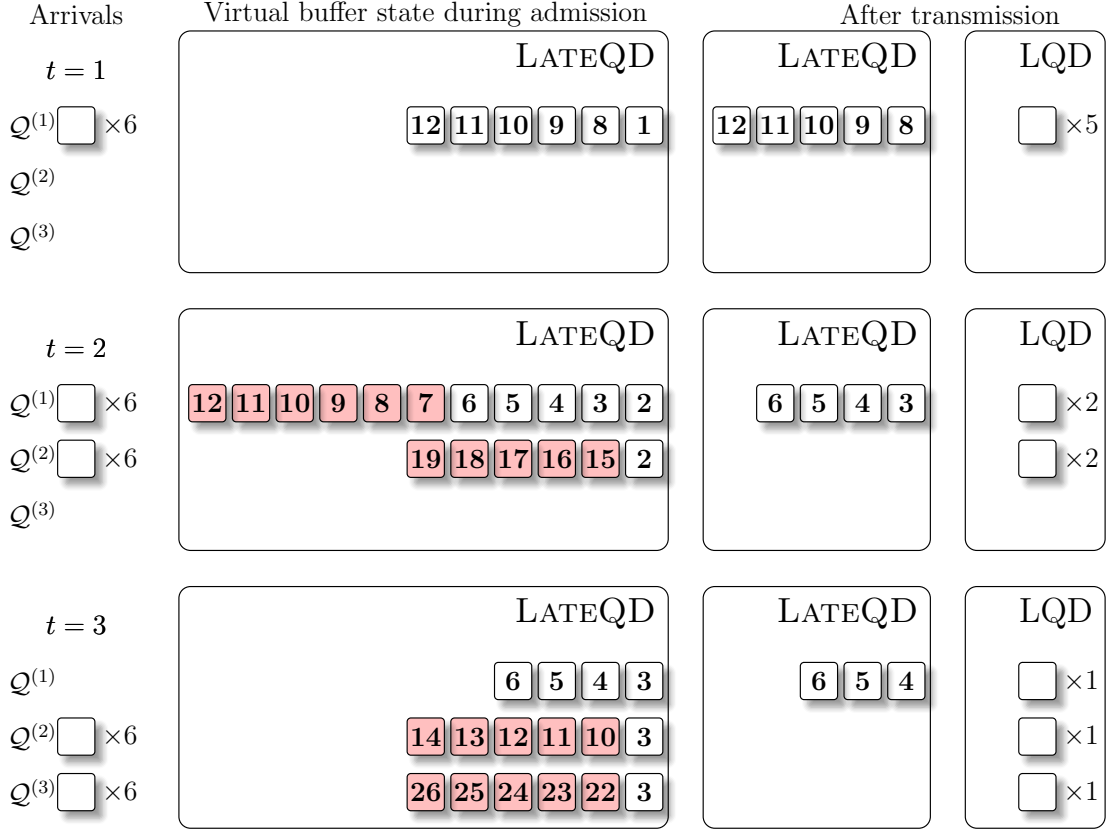
\begin{figure}[!t]\centering
\scalebox{\myscalefigtwo}{
\begin{tikzpicture}[x=1.2cm, y=1.2cm, scale=0.8]
    \tbltext{-5}{7}{$\Q 1$};
    \tbltext{-5}{6}{$\Q 2$};
    \tbltext{-5}{5}{$\Q 3$};
    \tbltext{-5}{2}{$\Q 1$};
    \tbltext{-5}{1}{$\Q 2$};
    \tbltext{-5}{0}{$\Q 3$};
    \tbltext{-5}{12}{$\Q 1$};
    \tbltext{-5}{11}{$\Q 2$};
    \tbltext{-5}{10}{$\Q 3$};

	\tbltext{-4}{14}{{Arrivals}};
    \tbltext{4}{14}{{Virtual buffer state during admission}};
    \tbltext{16}{14}{{After transmission}};

    \tblp{-4}{12}{color07}{};
	\tbltext{-3}{12}{$\times 6$};
	
	\myboxessetup{1}{10}{13};
	\myboxessetup{2}{5}{8};
	\myboxessetup{3}{0}{3};

    \foreach \ax in {8,...,12} {
      \pgfmathsetmacro\tx{\myxvirt-\ax+7}
      \tblp{\tx}{12}{color07}{\ax};
    }
    \tblp{\myxvirt}{12}{color07}{1};

    \foreach \ax in {8,...,12} {
      \pgfmathsetmacro\tx{\myxafter-\ax+8}
      \tblp{\tx}{12}{color07}{\ax};
    }

    \tblp{\myxlqd-1}{12}{color07}{};
    \tbltext{\myxlqd}{12}{$\times 5$};

    \tbltext{-4}{13}{{\bf $t=1$}};
    \tbltext{-4}{8}{{\bf $t=2$}};
    \tblp{-4}{7}{color07}{};
	\tbltext{-3}{7}{$\times 6$};
    \tblp{-4}{6}{color07}{};
	\tbltext{-3}{6}{$\times 6$};

    \tbltext{-4}{3}{{\bf $t=3$}};
    \tblp{-4}{1}{color07}{};
	\tbltext{-3}{1}{$\times 6$};
    \tblp{-4}{0}{color07}{};
	\tbltext{-3}{0}{$\times 6$};

    \foreach \ax in {2,...,6} {
      \pgfmathsetmacro\tx{\myxvirt-\ax+2}
      \tblp{\tx}{7}{color07}{\ax};
    }

    \foreach \ax in {7,...,12} {
      \pgfmathsetmacro\tx{\myxvirt-\ax+2}
      \tblp{\tx}{7}{color02}{\ax};
    }

    \foreach \ax in {15,...,19} {
      \pgfmathsetmacro\tx{\myxvirt-\ax+14}
      \tblp{\tx}{6}{color02}{\ax};
    }
    \tblp{\myxvirt}{6}{color07}{2};

    \foreach \ax in {3,...,6} {
      \pgfmathsetmacro\tx{\myxafter-\ax+3}
      \tblp{\tx}{7}{color07}{\ax};
    }

    \tblp{\myxlqd-1}{7}{color07}{};
    \tbltext{\myxlqd}{7}{$\times 2$};
    \tblp{\myxlqd-1}{6}{color07}{};
    \tbltext{\myxlqd}{6}{$\times 2$};

    \foreach \ax in {3,...,6} {
      \pgfmathsetmacro\tx{\myxvirt-\ax+3}
      \tblp{\tx}{2}{color07}{\ax};
    }

    \foreach \ax in {10,...,14} {
      \pgfmathsetmacro\tx{\myxvirt-\ax+9}
      \tblp{\tx}{1}{color02}{\ax};
    }
    \tblp{\myxvirt}{1}{color07}{3};

    \foreach \ax in {22,...,26} {
      \pgfmathsetmacro\tx{\myxvirt-\ax+21}
      \tblp{\tx}{0}{color02}{\ax};
    }
    \tblp{\myxvirt}{0}{color07}{3};

    \foreach \ax in {4,...,6} {
      \pgfmathsetmacro\tx{\myxafter-\ax+4}
      \tblp{\tx}{2}{color07}{\ax};
    }

    \tblp{\myxlqd-1}{2}{color07}{};
    \tbltext{\myxlqd}{2}{$\times 1$};
    \tblp{\myxlqd-1}{1}{color07}{};
    \tbltext{\myxlqd}{1}{$\times 1$};
    \tblp{\myxlqd-1}{0}{color07}{};
    \tbltext{\myxlqd}{0}{$\times 1$};
\end{tikzpicture}
}
    \caption{Sample operation of $\oopt$ for $N=3$, $B=6$. Numbers on packets show provisional transmission moments for an algorithm with infinite buffer. $\Q 1$ receives $B$ packets on $t\in\{1,2\}$, $\Q 2$ receives $B$ packets on $t\in\{2,3,4\}$, $\Q 3$ receives $B$ packets on $t\in\{3,4,5,6\}$. Packets shaded in red are pushed out. The right column shows $\LQD$ buffer state after transmission on the same input.}
    \label{fig:opt}
\end{figure}

%% file: extensions.tex
\section{Model extensions}\label{sec:extensions}

\subsection{Fractional packets}\label{sec:frac}

We begin our extensions to the basic model by considering \emph{fractional packets}. In this setting, we allow algorithms to drop and store in queues fractions of a packet, thus making memory constraints continuous (but the time is still discrete and slotted and arriving packets are still integer). This means that now an algorithm is allowed, e.g., to push out half a packet from queue $1$ and half a packet from queue $2$ when the buffer overflows by one packet. This is an important extension specifically for LQD since now it will push out fractions of packets in cases when there are several longest queues.

Although this significantly extends the space of algorithms, next we show that the optimal algorithm can be left unchanged.

\begin{theorem}\label{thm:frac}
If $B$ is integer, in the setting with fractional packets there exists an optimal (clairvoyant) algorithm $\OPT$ that does not use fractional packets.
\end{theorem}
\begin{proof}

Without loss of generality, we may assume that $\OPT$ does not push out already accepted packets.

Fix an input sequence $\tau$. Consider the operation of an arbitrary optimal algorithm that uses the smallest number of fractional packets (out of all optimal algorithms) on input $\tau$. 
First let us show that at every moment of time there can be at most one fractional packet in the buffer. Assume the opposite and choose the first time slot $t$ when two fractional packets reside in memory at the same time. Let us now increase the fractional packet that will be transmitted earlier (by assumption, they cannot be pushed out) and decrease the one that will be transmitted later so that one of them becomes integer, resolving ties arbitrarily. This operation can leave the rest of the algorithm's operation unchanged because it only decreases the total memory requirements on later time slots and increases on previous time slots only to the effect of rounding it up, which cannot violate an integer total memory constraint $B$. Moreover, it decreases the total number of fractional packets, which contradicts our assumption.

Now at every time moment at most one fractional packet is stored in memory. But since the global memory constraint $B$ is integer, we can round the sizes of all packets up without violating this constraint.
\end{proof}

\begin{corollary}
The algorithm $\oopt$ (Definition~\ref{def:oopt}) is also optimal for the setting with fractional packets.
\end{corollary}

\subsection{Structure of hard instances}\label{sec:hard}

The next step is to introduce the special form of adversarial input sequences that we use to construct hard instances for $\LQD$ and general online algorithms. In what follows we only consider infinite instances for which every queue $\Q i$ has a finite (and possibly empty) time interval $[b_i, e_i]$ such that packets arrive to $\Q i$ only during this interval, and on every $t\in[b_i,e_i]$ exactly $B$ packets arrive to $\Q i$. This special form of input instances proves to be sufficient for our constructions of lower bounds. Note that an infinite instance can potentially use an infinite number of queues, but if the intervals $[b_i,e_i]$ are finite for every $\Q i$ it can always be emulated by reusing queues that have become dead more than $B$ time slots ago.

\begin{definition}
We call a queue $\Q i$ with time interval $[b_i, e_i]$ \emph{dead} at time $t$ if $t> e_i$, \emph{dying} at time $t=e_i$, and \emph{live} if $t\in[b_i,e_i)$; we call $[b_i,e_i)$ the \emph{live interval} for queue $\Q i$. Note that before $b_i$ a queue is neither live nor dead (the metaphor is that it has not been born yet).
\end{definition}

\begin{proposition}\label{prop:opt}
On an infinite instance $\tau$ that has the above structure, $\oopt$ transmits the same set of packets as the algorithm $\ALG$ that operates as follows:
    \begin{itemize}
        \item after receiving new inputs, if the number of nonempty queues exceeds available memory, $\ALG$ takes one packet each from any $B$ queues (and its buffer becomes empty after transmission);
        \item otherwise, $\ALG$ takes one packet each from every live queue and operates equivalently to $\LQD$ on all dead and dying queues.
    \end{itemize}
\end{proposition}
\begin{proof}
For the first item, note that the $\oopt$ priority (expected transmission time under infinite memory) is smallest for the first packet in every queue, so all of them will be chosen, up to $B$ in total. The priority of all other packets in a live queue is so large
that $\oopt$ does not accept them ($B$ packets with better priorities will arrive at the next time slot), and for dead and dying queues the priorities of all packets increase sequentially, so on dead and dying queues $\ALG$ and $\oopt$ operate as $\LQD$. The only difference between $\ALG$ and $\oopt$ is in the case when no queue becomes dead on a given time slot: $\oopt$ will accept additional packets into live queues and $\ALG$ simply leaves a part of its memory free; however, $\oopt$ will obviously push out these additional packets on the very next time slot, and the overall set of transmitted packets is the same for $\oopt$ and $\ALG$.
\end{proof}


\newcommand\transm[2]{\mathrm{Trans}^{(#1)}_{#2}}

%% file: lqd.tex
\section{Simplified $\sqrt{2}$ lower bound for LQD}\label{sec:lqd}

Again, we consider only infinite instances introduced in Section~\ref{sec:hard}: each $\Q i$ has a finite (possibly empty) time interval $[b_i, e_i]$ such that packets arrive to $\Q i$ only during this interval and also possibly at time $0$, and on every $t\in[b_i,e_i]$ exactly $B$ packets arrive to $\Q i$. 


\begin{theorem}\label{SQRT2}
The competitive ratio of $\LQD$ is at least $\sqrt{2}$.
\end{theorem}
\begin{proof}
Let $h=\max\map{w \in \mathbb{N}}{a w + \frac{w (w + 1)}{2} \leq B}$, where $a$ is a constant to be defined later.
    
Consider the following instance: at time $0$ the queue $\Q j$ receives $j$ packets for $j < h$ and $B$ packets for $h \leq j \leq h + a$. After that, $\Q j$ is live (receives $B$ packets) during time interval $\left[\max (0, j - h - a), \max (-1, j - h))\right]$ and does not receive anything else.
    

    \begin{lemma}\label{lem1}
        For the $\LQD$ algorithm, the queue $\Q j$ at time moment $t$ contains $|\Q j_t|=\max (0, h - (t - j))$ packets if $j \leq t + h$ and either $h$ or $h + 1$ packets if $t + h < j \leq t + h + a$.
    \end{lemma}
    \begin{proof}
    We prove the lemma by induction. For $t = 0$ the statement is obvious. Suppose that for $t = t_0 - 1$ the statement holds, and let us show it for $t = t_0$. Indeed, one packet will be transmitted from every queue, queues $\Q j$ with $j < t_0 + h$ do not receive any more packets, and queues $\Q j$ with $t_0 + h \leq j \leq t_0 + h + a$ receive $B$ packets each, with $\LQD$ equalizing live queues and leaving either $h$ or $h + 1$ packets in each. The queue $\Q{t_0 + h}$ will have exactly $h$ packets by the definition of $h$. 
    \end{proof}
    
Thus, at every time moment $\LQD$ has $a + h$ nonempty queues and thus sends out $a + h$ packets. The operation of $\LQD$ is schematically illustrated on Figure~\ref{fig:lqd}. The shaded polygon shows queue occupancy at time moment $t$: queues from $\Q t$ to $\Q{t + h - 1}$ are dead and are slowly depleting, $\Q{t+h}$ is dying, and queues from $\Q{t + h+1}$ to $\Q{t + h + a}$ are live and hold $h$ or $h + 1$ packets each.

Consider an offline algorithm $\ALG$ that never pushes out accepted packets, accepts $p = \max\map{w}{\frac{w (w + 1)}{2} \leq B - a}$ packets from a dying queue (this requires clairvoyancy) and leaves $1$ packet in every live queue. Obviously, starting from time moment $p$ $\ALG$ transmits $a + p$ packets on every time slot. Note that $\floor*{\sqrt{2 (B - a)}} - 1 \leq p \leq \sqrt{2 (B - a)}$.

The operation of $\ALG$ is schematically illustrated on Figure~\ref{fig:alg}. Again, the shaded polygon shows queue occupancy at time $t$: $\ALG$ leaves only one packet in each live queue and tries to hold on to as many packets from dead and dying queues (from $\Q t$ to $\Q{t+p}$) as possible.
    
\begin{figure}[t]
    \centering
    \includegraphics[width=.8\linewidth]{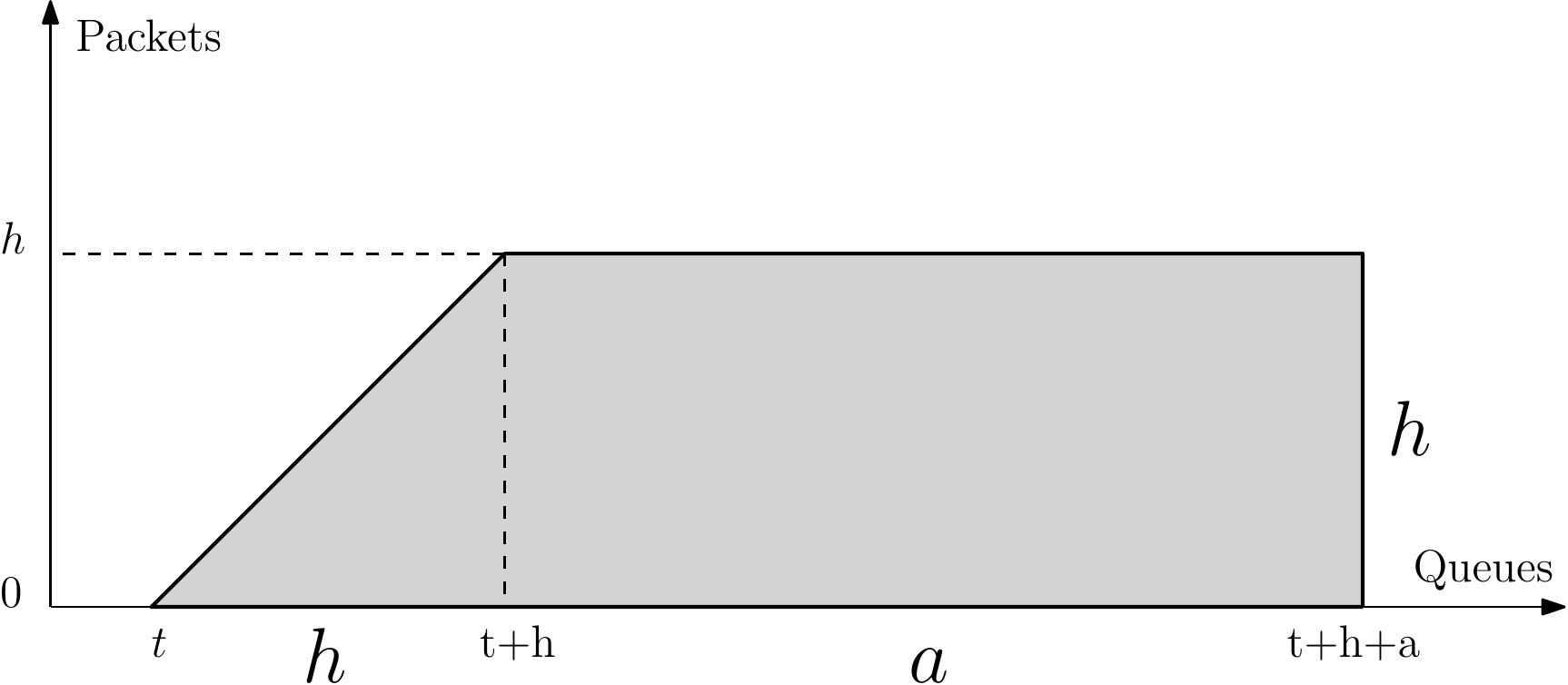}
    \caption{Queue occupancy for $\LQD$ at time moment $t$. The shaded polygon shows the number of packets in nonempty queues. $\LQD$ tries to equalize queue size as much as possible.}\label{fig:lqd}
\end{figure}
    
\begin{figure}[t]\centering
    \includegraphics[width=.8\linewidth]{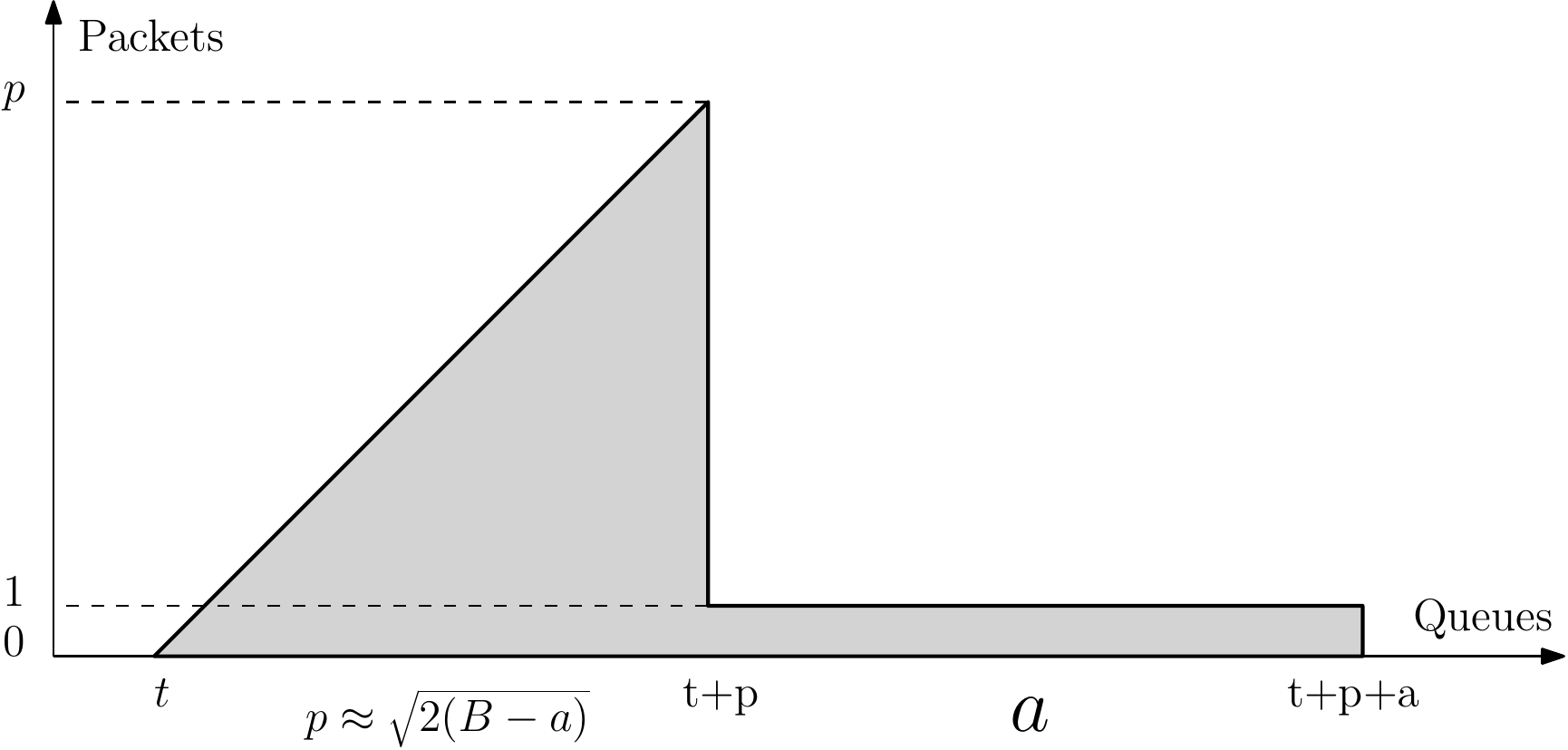}
    \caption{Queue occupancy for $\ALG$ at time moment $t$. $\ALG$ leaves one packet in every live queue and tries to preserve as many packets from dying queues as possible.}\label{fig:alg}
\end{figure}

Let $a = C\sqrt{B}$. Then
$$
        h = \max\map{w}{a w + \frac{w (w + 1)}{2} \leq B} = \floor*{\frac{-(2 a + 1) + \sqrt{(2 a + 1)^2 + 8B}}{2}},$$
        and we can find the ratio $\lim_{B \to \infty} \frac{a + p}{a + h}$ of the number of packets transmitted by $\ALG$ and $\LQD$ respectively:
    \begin{multline*}
            \lim_{B \to \infty} \frac{a + p}{a + h} = \lim_{B \to \infty} \frac{a + \floor{\sqrt{2 (B - a)}} - 1}{-\frac12 + \sqrt{(a + \frac12)^2 + 2B}} = \lim_{B \to \infty} \frac{a + \sqrt{2(B - a)}}{\sqrt{a^2 + 2B + O (a)}} = \\
             = \lim_{B \to \infty} \frac{a + \sqrt{2 (B - a)}}{\sqrt{a^2 + 2B}} = \lim_{B \to \infty} \frac{(C + \sqrt{2}) \sqrt{B}}{\sqrt{(C^2 + 2) B}} = \frac{C + \sqrt{2}}{\sqrt{C^2 + 2}}.
    \end{multline*}
    
Now $\frac{C + \sqrt{2}}{\sqrt{C^2 + 2}}$ achieves maximum at $\sqrt2$ for $C = \sqrt2$. 
It remains to show that we can construct a sequence of tests such that the competitiveness of $\LQD$ tends to $\sqrt{2}$ on this sequence. To do that, it suffices to construct a sequence of $(a_i, B_i)_{i=1}^{\infty}$ such that $\frac{a_i}{\sqrt{B_i}} \to \sqrt2$; this is obviously possible due to considerations above.
\end{proof}

%% file: general.tex
\section{General lower bound on online algorithms}\label{sec:general}

Throughout most of this section, we consider a model with discrete time and integer memory constraint but with fractional packets (Section~\ref{sec:frac}). This means that each queue $\Q i$, $i=1,\ldots,N$, is defined by its (possibly fractional) number of packets $|\Q i|$, on every transmission phase we update $|\Q i|:=\max(|\Q i|-1, 0)$, and the buffer size constraint also accounts for fractional packets, $\sum_{i=1}^N|\Q i| \le B$. We denote by $\LQD^f$ the fractional counterpart of $\LQD$ defined in exactly the same way.

\begin{proposition}\label{cor1}
    The competitive ratio of $\LQD^f$ is at least $\sqrt{2}$.
\end{proposition}
\begin{proof}
    Consider the instance defined in Theorem~\ref{SQRT2}. Note that if the problem parameters $h$, $a$, and $B$ are integers such that $B = a h + \frac{h (h + 1)}{2}$ exactly, with no rounding, then obviously, $\LQD = \LQD^f$ on this instance.
    
    It remains to show that there exists an increasing sequence of pairs of integers $(a_i, B_i)$ such that $\lim \frac{a_i}{\sqrt{B_i}} = \sqrt{2}$, and at the same time $B_i = a_i h + \frac{h (h + 1)}{2}$ for all $i$. For this, consider an arbitrary sequence of pairs satisfying the first condition, for each $(a_i, B_i)$ find $h_i = \map{w}{a w + \frac{w (w + 1)}{2}}$, and let $B_i^\prime = a h_i + \frac{h_i (h_i + 1)}{2}$. Obviously, $B_i - (a_i + h_i) \leq B_i^\prime \leq B_i$, and $a_i + h_i = o(\sqrt{B})$, so $\lim \frac{a_i}{\sqrt{B_i - (a_i + h_i)}} = \lim \frac{a_i}{\sqrt{B_i}} = \sqrt{2}$. Then $\frac{a_i}{\sqrt{B_i^\prime}} = \sqrt{2}$, so the sequence of pairs $(a_i, B_i^\prime)$ satisfies both conditions.
\end{proof}

\begin{theorem}
    The competitive ratio of any online algorithm in the model with fractional packets and discrete time is at least $\sqrt2$.
\end{theorem}
\begin{proof}
    Let $B$, $a$, and $h$ satisfy $B = a h + \frac{h (h + 1)}{2}$.
    Consider an arbitrary online algorithm $\ALG$. We consider the instance from Theorem~\ref{SQRT2} and modify it in an adversarial way such that at every time moment, there are exactly $a + 1$ live queues present, and at every time moment, the shortest live queue of the algorithm $\ALG$ dies. Note that the number of packets transmitted by $\LQD^f$ and by the offline algorithm defined in Theorem~\ref{SQRT2} does not change due to symmetry, so it suffices to show that $\ALG$ sends out no more packets than $\LQD^f$.
    
    \begin{lemma}\label{lem2}
    For any deterministic online algorithm $\ALG$, for the given adversarial example there exists an online algorithm $\ALG^\prime$ that transmits the same set of packets but never pushes packets out of dead queues \textup(except possibly at time moment $0$\textup).
    \end{lemma}
    \begin{proof}
        Consider a time moment $t$. Note that since the input instance is fixed, we can simulate the operation of $\ALG$ for $B$ steps ahead and find the set of packets from the shortest queue that will be pushed out. Then $\ALG^\prime$ pushes these packets out at time moment $t$. By pushing packets out in advance, we only shorten the queue further and do not increase the total memory used. As we only shorten the shortest queue, the input instance can be left unchanged.
    \end{proof}

Note that if there exists an online algorithm with competitive ratio less than $\sqrt{2}$, then there exists an algorithm that transmits all packets from the first $h - 1$ queues and has the same competitive ratio on the considered instance.

\begin{lemma}\label{lem3}
For every online algorithm $\ALG$ with competitive ratio $\alpha > 0$, there exists an algorithm that transmits all packets from the first $h - 1$ queues and has the same competitive ratio on every instance with an infinite number of packets.
\end{lemma}
\begin{proof}
Consider an algorithm $\ALG^\prime$ that during the first $h - 1$ steps transmits only packets from the first $h - 1$ queues, and then on every step accepts the same packets as $\ALG$. Obviously, the number of packets transmitted by $\ALG$ and $\ALG^\prime$ on any input instance differ by at most $(h - 1)B$, which is a constant that does not influence the competitive ratio.
\end{proof}
    
Thus, by a combination of Lemmas~\ref{lem2} and~\ref{lem3} we can assume that $\ALG$ never pushes packets out of dead queues, and the number of packets transmitted from dead queues can be computed as the sum of the sizes of shortest queues on every time step. Note that $\LQD^f$ always transmits from all live queues, so to establish that $\ALG$ transmits no more than $\LQD^f$ it suffices to prove that it does not transmit more packets from dead queues.

Let us restate the problem. Let $B$ be the buffer size, $a$ be the number of live queues, $D$ be the number of packets in dead queues, so $M=B-D$ is an upper bound on the number of packets in live queues, and let $B$, $a$, and $h$ be as in Proposition~\ref{cor1}, i.e., $a\approx \sqrt{2B}$. The number $M$ depends on the time step, and initially $M = B - \frac{h (h - 1)}{2}$ (all queues with packets are live in the first time slot). By the assumptions we have accumulated on $\ALG$, the only freedom it has left is to choose how many packets remain in the queue $\Q\ast$ that becomes dead at the current time slot: we have assumed that $\ALG$ does not push out from dead queues.
Thus, the decision $\ALG$ makes on every time step reduces to choosing a number $x \leq \frac{M}{a}$ (since $\Q\ast$ is the shortest of $a$ queues with at most $M$ packets in total), adding $x$ to the number of transmitted packets (all packets from dead queues get transmitted), but over the next $\floor{x}$ time steps reducing $M$ by $x - t$ on each time step $t=1,\ldots,\floor{x}$ since packets in $\Q\ast$ now take up space in the buffer. $\LQD^f$ chooses on every step $x = \frac{M}{a}$; let us show that this is the optimal choice.

Suppose that $\ALG$ has chosen $x = y$ for some $y < \frac Ma$ and let us try to replace this choice with $x = y + \eps$, where $0 < \eps < 1$. After this modification, $\ALG$ transmits $\eps$ more packets from the current dying queue $\Q\ast$, but over the next $\ceil{y + \eps} - 1$ steps the memory bound $M$ decreases by $\eps$. Since $M$ has decreased, this modification may influence other choices $\ALG$ makes down the line.

After we have changed $y$ to $y + \eps$, during $\ceil{y + \eps}-1$ steps we would like to have $\eps$ fewer packets in live queues (since our instance has an infinite number of packets we can assume that $\ALG$ uses up the entire buffer). Consider some time step $t$ from this time interval. There are two cases. If the shortest live queue in $\ALG$ has at least $\frac{\eps}{a}$ packets, we reduce every queue by $\frac{\eps}{a}$ packets, the shortest live queue remains shortest, and no further changes are needed.

If, on the other hand, the shortest live queue has $z<\frac{\eps}a$ packets we reduce this queue to zero and remove the other $\eps-z$ packets from other queues arbitrarily. To show that this is possible, we estimate the number of packets in dead queues from above: a queue that became dead $k$ time slots ago cannot have more than  $\ceil{\frac{B}{a}} - k - 1$ packets, so the total number of packets in dead queues does not exceed $\frac12{\ceil{\frac{B}{a}}(\ceil{\frac{B}{a}} - 1)}$. Since $a\approx\sqrt{2B}$, there are $\frac14 B+O(\sqrt{B})$ packets in dead queues, and hence enough packets in live queues.


After this change, over each of the next $\ceil{y + \eps} - 1$ steps the algorithm additionally loses no more than $\frac{\eps}{a}$ transmitted packets; recall again that we are counting only packets transmitted from dead queues. Since $\ceil{y + \eps} - 1\le\frac Ma$, we have lost at most $\frac {\eps M}{a^2}$ packets, and since $M<B$ and $a\approx\sqrt{2B}$ we have lost a number of packets that tends to $\frac\eps 2$ as $B$ increases, so this change is favorable for $\ALG$.
Later choices of $x$ do not change since choices that might influence $M$ on those steps have not increased it.

By this reasoning, we see that for every $0\le x < \frac{\eps}a$ it is favorable for the algorithm to choose a larger $x$, that is, $\ALG(\sigma)$ is an increasing function of $x$. 
Now let us consider the choices of $x$ an arbitrary online algorithm makes and replace each choice, one by one, with the maximal possible. For each such change, the number of transmitted packets does not decrease, therefore the number of transmitted packets after all transformations can only increase. On the other hand, the set of choices after all transformations will coincide with the set of choices $\LQD^f$ makes since $\LQD^f$ by construction makes the largest possible choice of $x$. Therefore, $\LQD^f$ is optimal.
\end{proof}

\begin{corollary}
    The competitive ratio of any online algorithm in the basic model with discrete packets is at least $\sqrt{2}$.
\end{corollary}
\begin{proof}
We compared the number of packets transmitted by an arbitrary online algorithm with the number of packets transmitted by the offline algorithm from Theorem~\ref{SQRT2}, which does not use fractional packets.
\end{proof}

%% file: linprog.tex
\section{Proving lower bounds by linear programming}\label{sec:linprog}

\subsection{Family of hard instances $\Phi_k$}

Our main instance for subsequent considerations is a family of infinite input instances $\Phi_k$ where queue $\Q{j}$ receives packets during the time interval $\left[k\floor{\frac{j-1}{k}}+1; j\right]$. We denote by $\transm{i}{[t,t']} (\A)$ the set of packets transmitted by an algorithm $\A$ from queue $\Q i$ over time interval $[t,t']$.

\begin{lemma}\label{LinearBound}
For the input $\Phi_k$, 
the number of packets transmitted by $\LQD$ from dead queue $\Q j$ is at most one greater than the number of packets transmitted by $\LQD$ from dead queue $\Q{j+1}$: $\left|\transm{j}{[e_j,\infty)} (\LQD)\right| \le \left|\transm{j+1}{[e_{j+1},\infty)} (\LQD)\right| + 1$.
\end{lemma}
\begin{proof}
The queue $\Q{j+1}$ was longest during time slot $j+1$, and before that moment queue $j$ was dead for only one time slot.
\end{proof}

Sample operation of our algorithms on $\Phi_k$ for $k=4$ and two incoming cycles ($8$ queues) is shown on Figure~\ref{fig:phi}. Fig.~\ref{fig:phi}a shows which queues are live at which time slots. Figs.~\ref{fig:phi}b and~\ref{fig:phi}c show the operation of $\LQD$ and $\oopt$ on this example: a shaded square means that the algorithm transmits from this queue at this time moment, and the number inside shows how many packets are left after transmission. The total number of transmitted packets in this example (total number of shaded squares on Fig.~\ref{fig:phi}) is $31$ for $\LQD$ and $32$ for $\oopt$.

\begin{figure}[!t]
\setlength{\tabcolsep}{-3pt}
\begin{tabular}{ccc}
\scalebox{\myscale}{
\begin{tikzpicture}
    \foreach \ay in {1,...,8} {
      \pgfmathsetmacro\ty{\ay-1}
      \tbltext{0}{\ty}{$\Q{\ay}$};
    }
    \tbltext{0}{-1}{{$\mathbf{t=}$}};
    \foreach \ax in {1,...,8} {
      \tbltext{\ax}{-1}{{\bf \ax}};
    }
    \foreach \ay in {1,...,4} {
        \foreach \ax in {\ay,...,4} { 
            \pgfmathsetmacro\tx{\ax-1}
            \tblp{\ay}{\tx}{color07}{};
            \pgfmathsetmacro\ty{\ay+4}
            \pgfmathsetmacro\tx{\ax+3}
            \tblp{\ty}{\tx}{color07}{};
        }
    }
\end{tikzpicture}
}
& 
\scalebox{\myscale}{
\begin{tikzpicture}
    \foreach \ax in {1,...,11} {
      \tbltext{\ax}{-1}{{\bf \ax}};
    }
    \tblp{1}{3}{color04}{0};
    \tblp{1}{2}{color04}{0};
    \tblp{1}{1}{color04}{1};
    \tblp{1}{0}{color04}{1};
    
    \tblp{2}{3}{color04}{0};
    \tblp{2}{2}{color04}{1};
    \tblp{2}{1}{color04}{1};
    \tblp{2}{0}{color04}{0};

    \tblp{3}{3}{color04}{1};
    \tblp{3}{2}{color04}{2};
    \tblp{3}{1}{color04}{0};

    \tblp{4}{3}{color04}{3};
    \tblp{4}{2}{color04}{1};
    
    \tblp{5}{7}{color04}{0};
    \tblp{5}{6}{color04}{0};
    \tblp{5}{5}{color04}{0};
    \tblp{5}{4}{color04}{0};
    \tblp{5}{3}{color04}{0};
    \tblp{5}{2}{color04}{0};
    
    \tblp{6}{7}{color04}{1};
    \tblp{6}{6}{color04}{1};
    \tblp{6}{5}{color04}{1};

    \tblp{7}{7}{color04}{1};
    \tblp{7}{6}{color04}{2};
    \tblp{7}{5}{color04}{0};

    \tblp{8}{7}{color04}{3};
    \tblp{8}{6}{color04}{1};
    
    \tblp{9}{7}{color04}{2};
    \tblp{9}{6}{color04}{0};

    \tblp{10}{7}{color04}{1};
    \tblp{11}{7}{color04}{0};
\end{tikzpicture}
}
&
\scalebox{\myscale}{
\begin{tikzpicture}
    \foreach \ax in {1,...,9} {
      \tbltext{\ax}{-1}{{\bf \ax}};
    }
    \tblp{1}{3}{color05}{0};
    \tblp{1}{2}{color05}{0};
    \tblp{1}{1}{color05}{0};
    \tblp{1}{0}{color05}{2};
    
    \tblp{2}{3}{color05}{0};
    \tblp{2}{2}{color05}{0};
    \tblp{2}{1}{color05}{1};
    \tblp{2}{0}{color05}{1};

    \tblp{3}{3}{color05}{0};
    \tblp{3}{2}{color05}{2};
    \tblp{3}{1}{color05}{0};
    \tblp{3}{0}{color05}{0};

    \tblp{4}{3}{color05}{3};
    \tblp{4}{2}{color05}{1};
    
    \tblp{5}{7}{color05}{0};
    \tblp{5}{6}{color05}{0};
    \tblp{5}{5}{color05}{0};
    \tblp{5}{4}{color05}{0};
    \tblp{5}{3}{color05}{0};
    \tblp{5}{2}{color05}{0};
    
    \tblp{6}{7}{color05}{0};
    \tblp{6}{6}{color05}{0};
    \tblp{6}{5}{color05}{3};

    \tblp{7}{7}{color05}{0};
    \tblp{7}{6}{color05}{1};
    \tblp{7}{5}{color05}{2};

    \tblp{8}{7}{color05}{1};
    \tblp{8}{6}{color05}{1};
    \tblp{8}{5}{color05}{1};

    \tblp{9}{7}{color05}{0};
    \tblp{9}{6}{color05}{0};
    \tblp{9}{5}{color05}{0};
\end{tikzpicture}
}
\\
{\scriptsize (a) $\Phi_k$ for $k=4$}
& 
{\scriptsize (b) $\LQD$ operation, $B=6$}
& 
{\scriptsize (c) $\oopt$ operation, $B=6$}
\\
\end{tabular}
\caption{Sample operation on the $\Phi_k$ input instance for $k=4$ and two input cycles: (a) live queues in $\Phi_k$; (b) $\LQD$ operation on this instance for $B=6$; (c) $\oopt$ operation for $B=6$. Numbers show the number of packets left in a queue after transmission.}\label{fig:phi}
\end{figure}
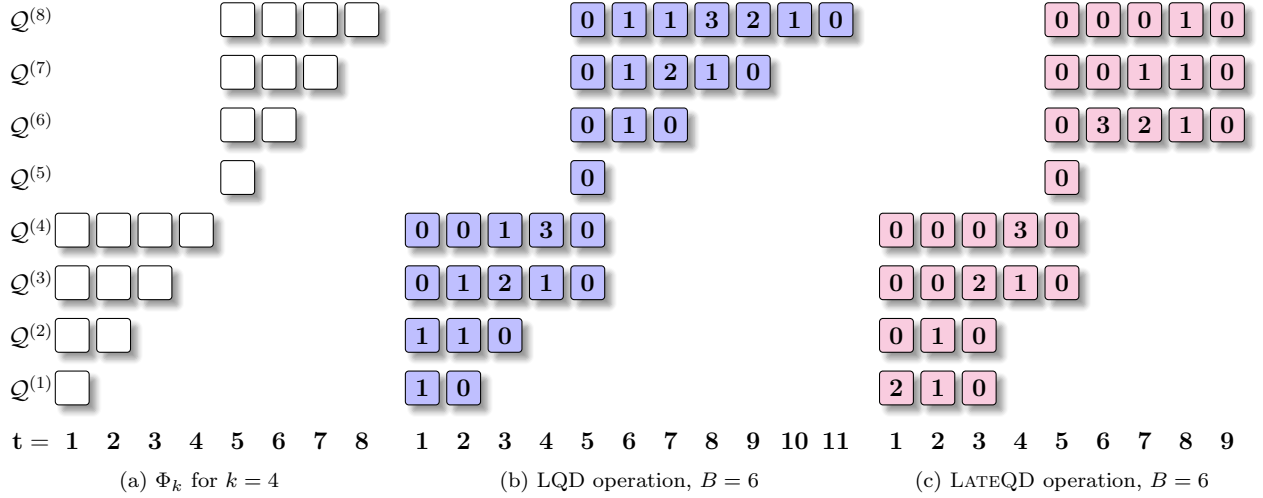

\subsection{Reduction to linear programming}

The explicit form of hard instances introduced above has made the search for lower bounds amenable to computer simulations. In this section, we prove several properties of the behaviour of online algorithms on this special case of hard instances and turn the search for a lower bound into a linear programming problem. Then we solve these problems in an experimental simulation and obtain a lower bound on the competitiveness of $\LQD$ that exceeds $\sqrt{2}$. Moreover, in the same vein we formulate such properties for any deterministic online algorithm and show a general lower bound that exceeds $\frac 43$ shown in~\cite{HK+01,AKM08}.

We again consider instances of the form $\Phi_k$ where queue $\Q{j}$ receives packets during the time interval $\left[k\floor{\frac{j-1}{k}}+1; j\right]$. We denote by $b_t$ the number of live queues which are not dying at time moment $t$.

\begin{theorem}\label{thm:lp}
On an input instance $\Phi_k$, denote by $a_t$ the number of packets sent from the queue that is dying (becomes dead) on time slot $t$. Then:
\begin{enumerate}[(1)]
    \item for any algorithm (in particular, $\oopt$), for every $t$
    $$b_t + \sum_{j\le t}\max\left(a_j-(t-j)+1,0\right) \le B;$$
    \item for any online algorithm, for every $t$
    $$b_ta_t + \sum_{j\le t}\max\left(a_j-(t-j)+1,0\right) \le B;$$
    \item for $\LQD$, moreover, 
    $a_t\le a_{t+1}+1;$
    \item the number of transmitted packets is equal to $\sum_t\left(a_t + b_t\right).$
\end{enumerate}
\end{theorem}
\begin{proof}
Constraint (1)  says that the total number of packets does not exceed the memory size $B$ at any particular timeslot. Constraint (2) results from the fact that live queues are indistinguishable, and an adversary can always choose the queue with the smallest number of packets to become dead, so the number of packets in a live queue for any online algorithm in the worst case is at least the same as the number of packets in a dying queue. Constraint (3) means that for $\LQD$ the number of packets in any dead queue cannot exceed the number of packets in a dying queue on any particular timeslot (they are the same). Finally, (4) specifies the objective function for the linear programming problem.
\end{proof}

Theorem~\ref{thm:lp} implies that a solution of the optimization problem with objective function defined in (4) and constraints defined in (1)-(3) is at least the number of packets transmitted by $\LQD$; without (3), transmitted by any online algorithm; without (2), transmitted by any clairvoyant algorithm. In the optimization problem, $\max$ can be replaced with an additional variable and two constraints so this is in fact a linear programming problem. Note the adversarial nature of the examples produced under constraints (1)-(2): the solution assumes that an adversary chooses which queue will die based on the algorithm's choices.

Note that for the optimal clairvoyant algorithm the upper bound produced by this optimization problem can be achieved: $\OPT$ can simply choose the resulting $a_t$ for its operation. For $\LQD$ and any online algorithm, the result remains an upper bound, so the ratio of the two solutions yields a lower bound on the competitiveness of $\LQD$ and an arbitrary online algorithm, respectively.

Formally, to obtain a solution we need to input an infinite instance $\Phi_k$, and even to get a good approximation we need to use a large number of repeating cycles from $\Phi_k$, resulting in an infeasible linear programming problem. We circumvent this problem by noting that:
\begin{itemize}
    \item first, terms in the sum $\sum_{j\le t}\max\left(a_j-(t-j)+1,0\right)$ are obviously zero for $j<t-B$;
    \item second, all inequalities still hold if we replace $a_j$ by the average of the sequence $\{a_j,a_{j-k},a_{j-2k},\ldots\}$; for ${\bar a}_j=\frac1{\lfloor j/k\rfloor+1}\sum_{s=0}^{\lfloor j/k\rfloor}a_{j-sk}$ we have
    \begin{align*}
        b_t + \sum_{j=t-B}^t\max\left({\bar a}_j-(t-j)+1,0\right) & \le B \text{ for any algorithm}; \\
        b_ta_t + \sum_{j=t-B}^t\max\left({\bar a}_j-(t-j)+1,0\right) &\le B\text{ for any online algorithm};\\
        {\bar a}_t &\le {\bar a}_{t+1}\text{ for LQD};
    \end{align*}
    this holds simply as a linear combination of linear inequalities;
    \item third, as $t$ grows ${\bar a}_t$ becomes closer and closer to ${\bar a}_{t-k}$, so for sufficiently large $t$ elements ${\bar a}_t$ become cyclic up to any predefined constant $\epsilon$;
    \item thus, we can replace an ``unrolled'' sequence of several cycles with a cyclic linear programming problem where we assume that $a_t=a_{t-k}$ and constrain the problem to $k$ variables; this works for the lower bound on competitiveness because the sequence of indices $t$ contains a subsequence where ${\bar a}_t$ converge for both $\LQD$ and the optimal algorithm; this sequence will yield the bound.
\end{itemize}

We find a solution for these linear programming problems using the COIN-OR linear solver~\cite{Ralphs13}. Our experiments showed, in particular, that on the $\Phi_k$ input instance for $k=300$ and $B = 27272$:
\begin{itemize}
    \item the optimal algorithm can transmit $114546$ packets;
    \item the number of packets transmitted by $\LQD$ is at most $79392$, giving competitiveness of $\LQD$ of at least $1.4427902$ (better than $\sqrt{2}$);
    \item an arbitrary online algorithm can transmit at most $86292$ packets, giving a general lower bound on the competitiveness of any online algorithm of $1.32742316$.
\end{itemize}

Note that the general lower bound above is worse even than $\frac 43$, let alone $\sqrt{2}$ we show in Section~\ref{sec:general}. This is due to the structure of the hard instances $\Phi_k$. The linear programming formalization shown above can also cover the bound in Section~\ref{sec:general} if we code in the exact hard instance from Theorem~\ref{SQRT2} (in fact, linear programming virtually trivializes in this case). We have been able to find instances with intermediate bounds between $4/3$ and $\sqrt{2}$ but have not found anything better than $\sqrt{2}$; whether this is possible remains to be seen in future work.

The linear programming approach has two advantages over simple empirical evaluation that will yield a better lower bound for $\LQD$ in the next section. First, here we are actually estimating the (approximate) competitiveness on an infinite input instance. Second, the linear programming approach allows to obtain general lower bounds for an arbitrary deterministic online algorithm (although in our experiments we have not improved upon known results in this way). As for the final constant for $\LQD$, in the next section we will improve on it, but only because instances that are feasible for linear programming are much smaller than for direct simulation.

%% file: eval.tex
\section{Empirical evaluation}\label{sec:eval}

\pgfplotsset{every axis/.append style={
    width=.5\linewidth,
    height={.35\linewidth},
    grid=major,
    legend pos=north west, legend style={font=\footnotesize},
    legend cell align={left},
    grid style={dashed,gray!25},
    y label style={at={(axis description cs:0.15,.5)}, anchor=south},
    ylabel near ticks, xlabel near ticks
    } }
\pgfplotsset{every tick label/.append style={font=\tiny}}

\def\mylinewidth{.65pt}
\def\mythicklinewidth{.95pt}

\newcommand{\mylineplotnm}[4]{\addplot[color=#2,line width=\mythicklinewidth,#3] table[x index=0,y index=#4] {plots/#1.csv};}

\newcommand{\mylineplot}[6]{\addplot[color=#2,line width=\mylinewidth,#3,mark=#4,mark phase=#5,mark repeat=2,mark size=1.3pt,mark options={solid}] table[x index=0,y index=#6] {plots/#1.csv};}

\newcommand{\plotsrpt}[1]{\mylineplot{#1}{black!100}{solid}{o}{1}{1}}
\newcommand{\plotfauxsrpt}[1]{\mylineplot{#1}{black!50}{solid}{square*}{0}{3}}
\newcommand{\plotlfifo}[1]{\mylineplot{#1}{black!50}{solid}{oplus}{0}{4}}
\newcommand{\plotlfair}[1]{\mylineplot{#1}{black!30}{solid}{triangle}{0}{5}}
\newcommand{\plotlheur}[1]{\mylineplot{#1}{black!30}{densely dashed}{diamond}{0}{6}}
\newcommand{\plotlpq}[1]{\mylineplot{#1}{black!50}{densely dashed}{otimes}{0}{7}}
\newcommand{\plotfifo}[1]{\mylineplot{#1}{black!100}{dotted}{diamond*}{0}{8}}
\newcommand{\plotfair}[1]{\mylineplot{#1}{black!100}{dash dot}{triangle*}{0}{9}}
\newcommand{\plotloop}[1]{\mylineplot{#1}{black!100}{densely dashed}{square}{1}{10}}

\newcommand{\mylegend}[0]{
    \begin{tikzpicture}[scale=.95]
    \node[anchor=west] (opt) at (1,1.0) {{\scriptsize \ref{gr:srpt} SRPT}};
    \node[anchor=west] (opt) at (2.5,1.0) {{\scriptsize \ref{gr:faux} Heuristic SRPT}};
    \node[anchor=west] (opt) at (5,1.0) {{\scriptsize \ref{gr:lfifo} \lfifo}};
    \node[anchor=west] (opt) at (7.5,1.0) {{\scriptsize \ref{gr:fifo} FIFO}};
    \node[anchor=west] (opt) at (10,1.0) {{\scriptsize \ref{gr:fair} FAIR}};
    \node[anchor=west] (opt) at (12.5,1.0) {{\scriptsize \ref{gr:loop} LOOP}};
    \end{tikzpicture}\vspace{-0.1cm}
}

\begin{figure}[!t]
    \centering
    \begin{tabular}{cc}
	\begin{tikzpicture}
      \begin{axis}[
          xlabel={\footnotesize (a) Parameter $k^2/B$ for $k=300$}, 
          ylabel={\footnotesize Competitive ratio}, 
              ymin=1.44, ymax=1.45,
    ytick={1.44,1.442,1.444,1.446,1.448,1.45},
    yticklabels={1.440,1.442,1.444,1.446,1.448,1.450},
        ]
            \mylineplotnm{cr_b_300}{red!50!black}{solid}{1};
      \end{axis}
    \end{tikzpicture}
    &
	\begin{tikzpicture}
      \begin{axis}[
          xlabel={\footnotesize (b) Parameter $k^2/B$ for $k=3000$}, 
    ymin=1.44, ymax=1.45,
    ytick={1.44,1.442,1.444,1.446,1.448,1.45},
    yticklabels={1.440,1.442,1.444,1.446,1.448,1.450},
        ]
            \mylineplotnm{cr_b_3000}{red!50!black}{solid}{1};
      \end{axis}
    \end{tikzpicture}
    \\
	\begin{tikzpicture}
      \begin{axis}[
    ymin=1.44, ymax=1.45,
    ytick={1.44,1.442,1.444,1.446,1.448,1.45},
    yticklabels={1.440,1.442,1.444,1.446,1.448,1.450},
          xlabel={\footnotesize (c) Parameter $k^2/B$ for $k=30000$}, 
          ylabel={\footnotesize Competitive ratio}, 
        ]
            \mylineplotnm{bound1}{red!50!black}{solid}{1};
      \end{axis}
    \end{tikzpicture}
    &
	\begin{tikzpicture}
      \begin{axis}[
          xlabel={\footnotesize (d) Parameter $k^2/B$ for $k=300000$}, 
    ymin=1.44, ymax=1.45,
    ytick={1.44,1.442,1.444,1.446,1.448,1.45},
    yticklabels={1.440,1.442,1.444,1.446,1.448,1.450},
        ]
            \mylineplotnm{cr_b_300000}{red!50!black}{solid}{1};
      \end{axis}
    \end{tikzpicture}
    \end{tabular}
    \caption{Experimental results: competitive ratio as a function of $k^2/B$ for different values of $k$.}
    \label{fig:compete}\vspace{.5cm}

    \centering
	\begin{tikzpicture}
      \begin{axis}[
          width=\linewidth, height=.4\linewidth,
          xlabel={\footnotesize Time slot $t$ (dying queue $\Q t$)}, 
          ylabel={\footnotesize Packets from a dying queue}, 
          ymin=0, ymax=50000, xmin=0,xmax=89500,
          enlargelimits=false,
          ytick={0,10000,20000,30000,40000},
          yticklabels={0,10000,20000,30000,40000},
          xtick={1,20000,40000,60000,80000},
          xticklabels={1,20000,40000,60000,80000},
        ]
        \mylineplotnm{deadsize_time_LQD_30000}{red!80!black}{solid}{1};
        \mylineplotnm{deadsize_time_LQD_clamped_30000}{blue!80!black}{solid}{1};
        \mylineplotnm{deadsize_time_OPT_30000}{green!50!black}{solid}{1};
        \legend{{$\LQD$},{$\LQD$, Lemma~\ref{LinearBound}},{$\oopt$}};
      \end{axis}
    \end{tikzpicture}
    \caption{Number of packets accepted and transmitted from dying queues on three cycles of $\Phi_k$ for $k=30000$.}\label{fig:dead}
\end{figure}

In Section~\ref{sec:opt}, we have defined a relatively simple and straightforward construction for a clairvoyant optimal algorithm. Moreover, in Section~\ref{sec:hard} we have introduced a specific form of input instances where we are looking for hard instances. In this section, we show how to find the total number of transmitted packets for the optimal algorithm on an input instance of this form, which makes it possible to run computer simulations in search for a lower bound on the competitiveness ratio of $\LQD$ or any other online algorithm.

We have run on $\Phi_k$ the optimal algorithm $\oopt$ as defined in Proposition~\ref{prop:opt} and $\LQD$ and measured the resulting competitiveness values. Figure~\ref{fig:compete} shows the plot of the competitiveness for several different values of $k$. A reasonable value of $B$ grows approximately as $k^2$ because we would like the sizes of nonempty queues to have the same order as their number, which on $\Phi_k$ is always of the order of $k$; therefore, we show the results as a function of $k^2/B$. We see that for instances $\Phi_k$, the resulting competitiveness approaches a concave function with a single maximum. The value of $k^2/B$ for maximum competitiveness grows with $k$ but throughout our simulations has always remained in the interval $[3,4]$; finding how this maximum depends on $k$ remains an interesting open problem for future study. The largest competitiveness value we have achieved in this experiment is $>1.44546086$ for $k=300000$ and $k^2/B=3.6$, i.e., $B=2.5\cdot 10^{10}$. Note that this result is not just a suggestive experiment but also constitutes a proof of the lower bound since it corresponds to a specific hard example.

Figure~\ref{fig:dead} reflects important properties of the considered algorithms. It shows the plots of number of packets accepted from a given queue at the time slot when it becomes dead, for three cycles of $\Phi_k$ for $k=30000$. It shows three different plots: number of such packets for $\oopt$ (the virtually horizontal green line), for $\LQD$ (the red line with a large maximum at the end of a cycle), and for $\LQD$ estimated with the bound from Lemma~\ref{LinearBound} (the blue line). The latter reflects the number of packets transmitted from each queue after it became dead (up to $\pm1$ due to rounding); that is, the blue line always remains below the red line. The $\LQD$ plot grows so much because the size of each dying queue is the ratio of the amount of free memory (equal to the number of nonempty queues, which is of the order $k$) to the number of live queues, which decreases to $1$ at the end of an input cycle. Bottom points of the red and blue lines correspond to beginnings of a new cycle;
all queues that at the end of a cycle have the estimate from Lemma~\ref{LinearBound} less than the number of accepted packets will have a size equal to this minimum up to rounding. The breaking point of the blue line happens after the number of steps equal to this minimum; it occurs because all such queues end at the same time for $\LQD$. On the last cycle of the input instance, the bound from Lemma~\ref{LinearBound} becomes trivial, and blue and red lines coincide. The plot of packets sent by $\oopt$ from a dead queue is not shown because it almost exactly coincides with the green line; on $\Phi_k$, $\oopt$ can only drop packets from dead queues at the beginning of a new cycle, and no more than $k$ packets from each.

We have made the source code for reproducing our experiments freely available on GitLab~\cite{GITLAB}.

%% file: conclusion.tex
\section{Conclusion}\label{sec:concl}

In this work, we have presented new results on the competitiveness of $\LQD$ and general online algorithms for the case of a shared memory switch with uniform packets, a longstanding open problem posed in~\cite{HK+01,AKM08,G10}. Our results are based on an explicit construction of the optimal clairvoyant algorithm, $\oopt$, and a generalized construction of a series of hard instances. With these constructions, we have proved a new general lower bound, have reduced finding the number of processed packets for online algorithms, including $\oopt$ and $\LQD$, to solving linear programming problems, and also have been able to implement the algorithms efficiently enough to test the introduced hard instances numerically. With these new techniques, we have been able to show a number of new lower bounds on competitiveness values:
\begin{itemize}
	\item a general lower bound of $\sqrt{2}$ for any deterministic online algorithm;
    \item with linear programming, we have shown that $\LQD$ is at least $1.4427902$-competitive (and the general lower bound of $\sqrt{2}$ for every online algorithm can also be obtained by linear programming);
    \item in numerical experiments, we have found that $\LQD$ is at least $1.44546086$-competitive.
\end{itemize}
These results improve upon the previously known lower bound of $\sqrt{2}$ for $\LQD$ and a general lower bound of $4/3$~\cite{HK+01,AKM08}. We hope that our approach can lead to even stronger results in further research.

%% file: main.bbl
\begin{thebibliography}{10}

\bibitem{AKM08}
William Aiello, Alexander Kesselman, and Yishay Mansour.
\newblock Competitive buffer management for shared-memory switches.
\newblock {\em ACM Transactions on Algorithms}, 5(1), 2008.

\bibitem{Al-Bawani2018}
Kamal Al-Bawani, Matthias Englert, and Matthias Westermann.
\newblock {Online Packet Scheduling for CIOQ and Buffered Crossbar Switches}.
\newblock {\em Algorithmica}, 80(12):3861--3888, 2018.

\bibitem{AMZ03}
Nir Andelman, Yishay Mansour, and An~Zhu.
\newblock Competitive queueing policies for {QoS} switches.
\newblock In {\em Proceedings of the $14^{\text{th}}$ Annual ACM-SIAM Symposium
  on Discrete Algorithms}, pages 761--770, 2003.

\bibitem{AL06}
Yossi Azar and Arik Litichevskey.
\newblock Maximizing throughput in multi-queue switches.
\newblock {\em Algorithmica}, 45(1):69--90, 2006.

\bibitem{AR05}
Yossi Azar and Yossi Richter.
\newblock Management of multi-queue switches in {QoS} networks.
\newblock {\em Algorithmica}, 43(1-2):81--96, 2005.

\bibitem{AzarR06}
Yossi Azar and Yossi Richter.
\newblock An improved algorithm for {CIOQ} switches.
\newblock {\em ACM Transactions on Algorithms}, 2(2):282--295, 2006.

\bibitem{GITLAB}
Ivan Bochkov, Alex Davydow, Nikita Gaevoy, and Sergey~I. Nikolenko.
\newblock Code for reproducing the experiments in ``{N}ew {C}ompetitiveness
  {B}ounds for the {S}hared {M}emory {S}witch''.
\newblock \url{https://gitlab.com/networking-lqd/networking-gcc}, 2019.

\bibitem{Borodin-ElYaniv}
Allan Borodin and Ran El-Yaniv.
\newblock {\em Online Computation and Competitive Analysis}.
\newblock Cambridge University Press, 1998.

\bibitem{CNK15}
Pavel Chuprikov, Sergey~I. Nikolenko, and Kirill Kogan.
\newblock Priority queueing with multiple packet characteristics.
\newblock In {\em 2015 IEEE International Conference on Computer
  Communications}, pages 1418--1426, 2015.

\bibitem{CNKD18}
Pavel Chuprikov, Sergey~I. Nikolenko, Kirill Kogan, and Alexei~P. Davydow.
\newblock Priority queueing for packets with two characteristics.
\newblock {\em IEEE Transactions on Networking}, 26(1):342--355, 2018.

\bibitem{DCNK17}
Alexei~P. Davydow, Pavel Chuprikov, Sergey~I. Nikolenko, and Kirill Kogan.
\newblock Throughput optimization with latency constraints.
\newblock In {\em 2017 IEEE International Conference on Computer
  Communications}, 2017.

\bibitem{ES04}
Leah Epstein and Rob van Stee.
\newblock Buffer management problems.
\newblock {\em SIGACT News}, 35(3):58--66, 2004.

\bibitem{EKNS14}
Patric Eugster, Kirill Kogan, Sergey~I. Nikolenko, and Alexander~V. Sirotkin.
\newblock Shared-memory buffer management for heterogeneous packet processing.
\newblock In {\em Proceedings of the $34^{\text{th}}$ International Conference
  on Distributed Computing Systems}, 2014.

\bibitem{G10}
Michael Goldwasser.
\newblock A survey of buffer management policies for packet switches.
\newblock {\em SIGACT News}, 41(1):100--128, 2010.

\bibitem{HK+01}
Ellen~L. Hahne, Alexander Kesselman, and Yishay Mansour.
\newblock Competitive buffer management for shared-memory switches.
\newblock In {\em SPAA}, pages 53--58, 2001.

\bibitem{KKM12}
Jun Kawahara, Koji Kobayashi, and Tomotaka Maeda.
\newblock Tight analysis of priority queuing policy for egress traffic.
\newblock {\em CoRR}, abs/1207.5959, 2012.

\bibitem{KKSS12}
Isaac Keslassy, Kirill Kogan, Gabriel Scalosub, and Michael Segal.
\newblock Providing performance guarantees in multipass network processors.
\newblock {\em IEEE/ACM Transactions of Networking}, 20(6):1895--1909, 2012.

\bibitem{KKM+10}
Alexander Kesselman, Kirill Kogan, and Michael Segal.
\newblock Packet mode and {QoS} algorithms for buffered crossbar switches with
  {FIFO} queuing.
\newblock {\em Distributed Computing}, 23(3):163--175, 2010.

\bibitem{KesselmanKS12}
Alexander Kesselman, Kirill Kogan, and Michael Segal.
\newblock Improved competitive performance bounds for {CIOQ} switches.
\newblock {\em Algorithmica}, 63(1-2):411--424, 2012.

\bibitem{KLM04}
Alexander Kesselman, Zvi Lotker, Yishay Mansour, Boaz Patt-Shamir, Baruch
  Schieber, and Maxim Sviridenko.
\newblock Buffer overflow management in {QoS} switches.
\newblock {\em SIAM J. Comput.}, 33(3):563--583, 2004.

\bibitem{KM04}
Alexander Kesselman and Yishay Mansour.
\newblock Harmonic buffer management policy for shared memory switches.
\newblock {\em Theor. Comput. Sci.}, 324(2-3):161--182, 2004.

\bibitem{KMS05}
Alexander Kesselman, Yishay Mansour, and Rob van Stee.
\newblock Improved competitive guarantees for {QoS} buffering.
\newblock {\em Algorithmica}, 43(1-2):63--80, 2005.

\bibitem{KesselmanR06}
Alexander Kesselman and Adi Ros{\'e}n.
\newblock Scheduling policies for {CIOQ} switches.
\newblock {\em Journal of Algorithms}, 60(1):60--83, 2006.

\bibitem{KesselmanR08}
Alexander Kesselman and Adi Ros{\'e}n.
\newblock Controlling {CIOQ} switches with priority queuing and in multistage
  interconnection networks.
\newblock {\em Journal of Interconnection Networks}, 9(1/2):53--72, 2008.

\bibitem{Kobayashi:2008:TBO:1522914.1522915}
Koji Kobayashi, Shuichi Miyazaki, and Yasuo Okabe.
\newblock A tight bound on online buffer management for two-port shared-memory
  switches.
\newblock {\em IEICE - Trans. Inf. Syst.}, E91-D(8):2105--2114, August 2008.

\bibitem{KMO09}
Koji Kobayashi, Shuichi Miyazaki, and Yasuo Okabe.
\newblock Competitive buffer management for multi-queue switches in {QoS}
  networks using packet buffering algorithms.
\newblock In {\em Proceedings of the $21^{\text{st}}$ ACM Symp. on Parallelism
  in Algorithms and Architectures (SPAA)}, pages 328--336, 2009.

\bibitem{KLNS13}
Kirill Kogan, Alejandro L{\'o}pez-Ortiz, Sergey~I. Nikolenko, and Alexander
  Sirotkin.
\newblock Multi-queued network processors for packets with heterogeneous
  processing requirements.
\newblock In {\em Proceedings of the $5^{\text{th}}$ International Conference
  on Communication Systems and Networks (COMSNETS 2013)}, pages 1--10, 2013.

\bibitem{KLNS12}
Kirill Kogan, Alejandro L\'opez-Ortiz, Sergey~I. Nikolenko, and Alexander~V.
  Sirotkin.
\newblock A taxonomy of semi-{FIFO} policies.
\newblock In {\em Proceedings of the $31^{\text{st}}$ IEEE International
  Performance Computing and Communications Conference (IPCCC 2012)}, pages
  295--304, 2012.

\bibitem{KLNS16}
Kirill Kogan, Alejandro L{\'o}pez-Ortiz, Sergey~I. Nikolenko, and Alexander~V.
  Sirotkin.
\newblock Online scheduling fifo policies with admission and push-out.
\newblock {\em Theory of Computing Systems}, 58(2):322--344, 2016.

\bibitem{KLNS+12}
Kirill Kogan, Alejandro L\'opez-Ortiz, Sergey~I. Nikolenko, Alexander~V.
  Sirotkin, and Denis Tugaryov.
\newblock {FIFO} queueing policies for packets with heterogeneous processing.
\newblock In {\em Proceedings of the $1^{\text{st}}$ Mediterranean Conference
  on Algorithms (MedAlg 2012), Lecture Notes in Computer Science}, volume 7659,
  pages 248--260, 2012.

\bibitem{NK14}
Sergey~I. Nikolenko and Kirill Kogan.
\newblock Single and multiple buffer processing.
\newblock In {\em Encyclopaedia of Algorithms}, pages 1--9. Springer, 2014.

\bibitem{Ralphs13}
T.~Ralphs.
\newblock An introduction to the coin-or optimization suite: Open source tools
  for building and solving optimization models.
\newblock In {\em Optimization Days, Montreal}, 2013.

\bibitem{Zhu04}
An~Zhu.
\newblock Analysis of queueing policies in {QoS} switches.
\newblock {\em Journal of Algorithms}, 53(2):137--168, 2004.

\end{thebibliography}
